\documentclass[11pt,letterpaper]{article}

\usepackage{fullpage}
\usepackage{amsthm,amsmath,amssymb}
\usepackage{algorithm, algorithmic}
\usepackage{mdframed}
\usepackage{graphicx}
\ifdefined\CR

\fi
\usepackage{enumitem}
\usepackage{xspace}
\usepackage{mdframed}
\usepackage{multirow}
\usepackage{hhline}

\ifdefined\CR
\newtheorem{theorem}{Theorem}
\else
\newtheorem{theorem}{Theorem}[section]
\fi
\newtheorem{lemma}[theorem]{Lemma}
\newtheorem{claim}[theorem]{Claim}
\newtheorem{definition}[theorem]{Definition}

\newtheorem{observation}[theorem]{Observation}

\newcommand{\set}[1]{\left\{#1\right\}}

\newcommand{\ceil}[1]{\left\lceil#1\right\rceil}

\DeclareMathOperator{\E}{\mathbb E}
\DeclareMathOperator{\union}{\bigcup}

\renewcommand{\tilde}{\widetilde}

\newcommand{\R}{\mathbb{R}}
\newcommand{\Z}{\mathbb{Z}}

\newcommand{\calA}{{\mathcal A}}

\newcommand{\calS}{{\mathcal S}}

\newcommand{\eps}{\epsilon}

\graphicspath{{figs/}}

\newcommand{\sfd}{\mathsf{d}}

\newcommand{\bffor}{{\bf for \xspace}}

\newcommand{\bfreturn}{{\bf return \xspace}}

\newcommand{\tprec}{\mathrm{prec}}
\newcommand{\PprecwC}{$P\big|\tprec\big|\sum_j w_jC_j$\xspace}
\newcommand{\PprecupwC}{$P \big| \tprec, p_j=1 \big| \sum_j w_jC_j$\xspace}

\begin{document}
	
	\title{Scheduling to Minimize Total Weighted Completion Time via Time-Indexed Linear Programming Relaxations}
	
	\author{Shi Li \thanks{Department of Computer Science and Engineering, University at Buffalo, Buffalo, NY, USA}}
	\date{}
	
	\maketitle
	\begin{abstract}
	We study approximation algorithms for scheduling problems with the objective of minimizing total weighted completion time, under identical and related machine models with job precedence constraints. We give algorithms that improve upon many previous 15 to 20-year-old state-of-art results. A major theme in these results is the use of time-indexed linear programming relaxations.  These are natural relaxations for their respective problems, but surprisingly are not studied in the literature. 
	
	We also consider the scheduling problem of minimizing total weighted completion time on unrelated machines. The recent breakthrough result of [Bansal-Srinivasan-Svensson, STOC 2016] gave a $(1.5-c)$-approximation for the problem, based on some lift-and-project SDP relaxation. Our main result is that a $(1.5 - c)$-approximation can also be achieved using a natural and considerably simpler time-indexed LP relaxation for the problem. We hope this relaxation can provide new insights into the problem.		
\end{abstract}
	

%
%
%

\section{Introduction}

	Scheduling jobs to minimize total weighted completion time is a well-studied topic in scheduling theory, operations research and approximation algorithms. A systematic study of this objective under many different machine models (e.g, identical, related and unrelated machine models,  job shop scheduling, precedence constraints, preemptions) was started in late 1990s and since then it has led to great progress on many fundamental scheduling problems. 
	
	In spite of these impressive results, the approximability of many problems is still poorly understood. Many of the state-of-art results that were developed in late 1990s or early 2000s have not been improved since then. Continuing the recent surge of interest on the total weighted completion time objective \cite{BSS16, IL16, Sku16}, 	we give improved approximation algorithms for many scheduling problems under this objective.  The machine models we study in this paper include identical machine model with job precedence constraints, with uniform and non-uniform job sizes, related machine model with job precedence constraints and unrelated machine model. 

	A major theme in our results is the use of time-indexed linear programming relaxations. Given the time aspect of scheduling problems, they are natural relaxations for their respective problems. 
	However, to the best of our knowledge, many of these relaxations were not studied in the literature and thus their power in deriving improved approximation ratios was not well-understood. Compared to other types of relaxations, 
	solutions to these relaxations give fractional scheduling of jobs on machines.  Many of our improved results were obtained by using the fractional scheduling to identify the loose analysis in previous results. 
	
	\subsection{Definitions of Problems and Our Results} We now formally describe the problems we study in the paper and state our results. In all of these problems, we have a set $J$ of $n$ jobs, a set $M$ of $m$ machines, each job $j \in J$ has a weight $w_j \in \Z_{> 0}$, and the objective to minimize is $\sum_{j \in J}w_j C_j$, where $C_j$ is the completion time of the job $j$. We consider non-preemptive schedules only.  So, a job must be processed on a machine without interruption. For simplicity, this global setting will not be repeated when we define problems.
	
	\paragraph{Scheduling on Identical Machines with Job Precedence Constraints}  In this problem, each job $j \in J$ has a processing time (or size) $p_j \in \Z_{>0}$. The $m$ machines are identical; each job $j$ must be scheduled on one of the $m$ machines  non-preemptively; namely, $j$ must be processed during a time interval of length $p_j$ on some machine. The completion time of $j$ is then the right endpoint of this interval. Each machine at any time can only process at most one job.  Moreover, there are precedence constraints given by a partial order ``$\prec$", where a constraint $j \prec j'$ requires that job $j'$ can only start after job $j$ is completed.  Using the popular three-field notation introduced by Graham et al.\ \cite{GLLR79}, this problem is described as \PprecwC. 
	
	For the related problem $P|\tprec|C_{\max}$, i.e, the problem with the same setting but with the makespan objective,  the seminal work of Graham \cite{Gra69} gives a $2$-approximation algorithm, based on a simple machine-driven list-scheduling algorithm. In the algorithm, the schedule is constructed in real-time.  As time goes, each idle machine shall pick any available job to process (a job is available if it is not scheduled but all its predecessors are completed.), if such a job exists; otherwise, it remains idle until some job becomes available.  
	%
%
	On the negative side, Lenstra and Rinnooy Kan \cite{LR78} proved a $(4/3-\epsilon)$-hardness of approximation for $P|\tprec|C_{\max}$. Under some stronger version of the Unique Game Conjecture (UGC) introduced by Bansal and Khot \cite{BK09}, Svensson \cite{Sve10} showed that $P|\tprec|C_{\max}$ is hard to approximate within a factor of $2-\epsilon$ for any $\epsilon > 0$.
	
	With precedence constraints, the weighted completion time objective is more general than makespan: one can create a dummy job of size $0$ and weight 1 that must be processed after all jobs in $J$, which have weight 0. Thus, the above negative results carry over to \PprecwC.  Indeed, Bansal and Khot \cite{BK09} showed that the problem with even one machine is already hard to approximate within a factor of $2-\epsilon$, under their stronger version of UGC.
	However, no better hardness results are known for \PprecwC, compared to those for $P|\tprec|C_{\max}$.

	On the positive side,  by combining the list-scheduling algorithm of Graham \cite{Gra69} with a convex programming relaxation for \PprecwC,  Hall et al.\ \cite{HSS97} gave a $7$-approximation for \PprecwC.  For the special case $1|\tprec|\sum_{j} w_j C_j$ of the problem where there is only 1 machine, Hall et al.\ \cite{HSS97} gave a $2$-approximation, which matches the $(2-\epsilon)$-hardness assuming the stronger version of UGC due to \cite{BK09}.  Later, Munier, Queyranne and Schulz (\cite{MQS98}, \cite{QS06}) gave the current best $4$-approximation algorithm for \PprecwC, using a convex programming relaxation similar to that in Hall et al.\ \cite{HSS97} and in Charkrabarti et al. \cite{CPS96}.  The convex programming gives a completion time vector $(C_j)_{j \in J}$, and  the algorithm of \cite{MQS98} runs a \emph{job-driven} list-scheduling algorithm using the order of jobs determined by the values $C_j - p_j/2$.  In the algorithm, we schedule jobs $j$ one by one, according to the non-increasing order of $C_j - p_j/2$; at any iteration, we schedule $j$ at an interval $(\tilde C_j - p_j, \tilde C_j]$ with the minimum $\tilde C_j$, subject to the precedence constraints and the $m$-machine  constraint.  It has been a long-standing open problem to improve this factor of $4$ (see the discussion after Open Problem 9 in \cite{SW99a}). 
	

	Munier, Queyranne and Schulz \cite{MQS98} also considered an important special case of the problem, denoted as \PprecupwC, in which all jobs have size $p_j = 1$. They showed that the approximation ratio of their algorithm becomes $3$ for the special case, which has not been improved since then. On the negative side, the $2-\eps$ strong UGC-hardness result of Bansal and Khot also applies to this special case.

	In this paper,  we improve the long-standing approximation ratios of 4 and 3 for \PprecwC and \PprecupwC due to Munier, Queyranne and Schulz \cite{MQS98, QS06}: 
	\begin{theorem}
		\label{theorem:main-PprecwC}
		There is a $2+2\ln2 + \epsilon < (3.387 + \epsilon)$-approximation algorithm for \PprecwC, for every $\epsilon > 0$.
	\end{theorem}
	
	\begin{theorem}
		\label{theorem:main-PprecwCuniform}
		There is a $1 + \sqrt{2} < 2.415$-approximation algorithm for \PprecupwC.			
	\end{theorem}

	\paragraph{Scheduling on Related Machines with Job Precedence Constraints} Then we consider the scheduling problem on related machines.  We have all the input parameters in the problem $P|\tprec|\sum_j w_jC_j$. Additionally, each machine $i \in M$ is given a speed $s_i > 0$ and the time of processing job $j$ on machine $i$ is $p_j/s_i$ (so the $m$ machines are not identical any more).  A job $j$ must be scheduled on some machine $i$ during an interval of length $p_j/s_i$.  Using the three-field notation, the problem is described as $Q|\tprec|\sum_j w_jC_j$.  
	
	Chudak and Shmoys \cite{CS97} gave the current best $O(\log m)$ approximation algorithm for the problem, improving upon the previous $O(\sqrt{m})$-approximation due to Jaffe \cite{Jaf80}. Using a general framework of Hall et al.\ \cite{HSS97} and Queyranne and Sviridenko \cite{QS02}, that converts an algorithm for a scheduling problem with makespan objective to an algorithm for the correspondent problem with weighted completion time objective,  Chudak and Shmoys reduced the problem $Q|\tprec|\sum_j w_jC_j$ to $Q|\tprec|C_{\max}$. 
	In their algorithm for $Q|\tprec|C_{\max}$, we partition the machines into groups, each containing machines of similar speeds. By solving an LP relaxation, we assign each job to a group of machines. Then we can run a generalization of the Graham's machine-driven list scheduling problem, that respect the job-to-group assignment.  The $O(\log m)$-factor comes from the number $O(\log m)$ of machine groups. 
	
	On the negative side, all the hardness results for $P|\tprec|C_{\max}$ carry over to both $Q|\tprec|C_{\max}$ and  $Q|\tprec|\sum_j w_jC_j$. Recently Bazzi and Norouzi-Fard \cite{BN15} showed that assuming the hardness of some optimization problem on $k$-partite graphs, both problems are hard to be approximated within any constant. 
	
	In this paper, we give a slightly better approximation ratio than $O(\log m)$ due to Chudak and Shmoys \cite{CS97}, for both  $Q|\tprec|C_{\max}$ and  $Q|\tprec|\sum_j w_jC_j$: 
	\begin{theorem}
		\label{theorem:main-QprecwC}
		There are $O(\log m/\log\log m)$-approximation algorithms for both $Q|\tprec|C_{\max}$ and $Q|\tprec|\sum_j w_jC_j$.
	\end{theorem}

	\paragraph{Scheduling on Unrelated Machines} Finally, we consider the classic scheduling problem to minimize total weighted completion time on unrelated machines (without precedence constraints). In this problem we are given a number $p_{i, j} \in \Z_{> 0}$ for every $i \in M, j \in J$,  indicating the time needed to process job $j$ on machine $i$.  
	This problem is denoted as $R||\sum_{j} w_j C_j$. 
	
	For this problem, there are many classic $3/2$-approximation algorithms, based on a weak time-indexed LP relaxation \cite{SS02} and a convex-programming relaxation (\cite{Sku01},  \cite{SS99}).  These algorithms are all based on independent rounding. Solving some LP (or convex programming) relaxation gives $y_{i, j}$ values, where each $y_{i, j}$ indicates the fraction of job $j$ that is assigned to machine $i$. Then the algorithms randomly and independently assign each job $j$ to a machine $i$, according to the distribution $\{y_{i, j}\}_i$.  Under this job-to-machine assignment, the optimum scheduling can be found by applying the Smith rule on individual machines.
	
	Improving the $3/2$-approximation ratio had been a long-standing open problem (see Open Problem 8 in \cite{SW99a}). The difficulty of improving the ratio comes from the fact that any independent rounding algorithm can not give a better than a $3/2$-approximation for $R||\sum_{j} w_j C_j$, as shown by Bansal, Srinivasan and Svensson \cite{BSS16}.  This lower bound is irrespective of the relaxation used: even if the fractional solution is already a convex combination of optimum integral schedules, independent rounding can only give a $3/2$-guarantee. To overcome this barrier, \cite{BSS16} introduced a novel dependence rounding scheme, which guarantees some strong negative correlation between events that jobs are assigned to the same machine $i$.  Combining this with their lifted SDP relaxation for the problem, Bansal, Srinivasan and Svensson gave a $(3/2-c)$-approximation algorithm for the problem $R||\sum_j w_j C_j$, where $c  = 1/(108 \times 20000)$. This solves the long-standing open problem in the affirmative. 

	
	Besides the slightly improved approximation ratio, our main contribution for this problem is that the $(1.5-c)$-approximation ratio can also be achieved using the following  natural time-indexed LP relaxation:
\ifdefined \CR
	\begin{equation}
	\min \qquad \sum_{j}w_j\sum_{i, s}x_{i, j, s}(s+p_{i,j}) \quad \text{s.t.} \tag{$\mathrm{LP}_{\mathrm{R}||\mathrm{wC}}$} \label{LP:RwC}
	\end{equation}
	\vspace*{-20pt}
	
	\begin{align}
		\sum_{i, s} x_{i, j, s} &= 1 & \quad \forall& j \label{LPC:RwC-job-j-scheduled} \\
		\sum_{j, s \in (t-p_{i, j}, t]} x_{i, j,s} &\leq 1 & \quad \forall& i, t \label{LPC:RwC-congestion} \\
		x_{i, j, s} & = 0 & \quad  \forall &i, j, s > T - p_{i, j}\label{LPC:RwC-job-length} \\
		x_{i, j, s} &\geq 0 & \quad \forall&i, j, s \label{LPC:RwC-non-negative}
	\end{align}
\else
	\begin{equation}
	\min \qquad \sum_{j}w_j\sum_{i, s}x_{i, j, s}(s+p_{i,j}) \qquad \qquad \text{s.t.} \tag{$\textsf{LP}_{\text{R}||\text{wC}}$} \label{LP:RwC}
	\end{equation}
	\vspace*{-30pt}
	
	\noindent
	\begin{minipage}{0.54\textwidth}
		\vspace*{15pt}
		\begin{align}
		\sum_{i, s} x_{i, j, s} &= 1 & \quad \forall& j \label{LPC:RwC-job-j-scheduled} \\[1pt]
		\sum_{j, s \in (t-p_{i, j}, t]} x_{i, j,s} &\leq 1 & \quad \forall& i, t \label{LPC:RwC-congestion} 
		\end{align}
	\end{minipage}
	\begin{minipage}{0.45\textwidth}
		\begin{align}
		x_{i, j, s} & = 0 & \quad  \forall &i, j, s > T - p_{i, j}\label{LPC:RwC-job-length} \\[17pt]
		x_{i, j, s} &\geq 0 & \quad \forall&i, j, s \label{LPC:RwC-non-negative}
		\end{align}
	\end{minipage}\medskip
\fi	

In the above LP, $T$ is a trivial upper bound on the makespan of any reasonable schedule ($T = \sum_{j}\max_{i: p_{i, j} \neq \infty}p_{i, j}$ suffices). $i, j,s$ and $t$ are restricted to elements in $M, J, \{0, 1, 2, \cdots, T-1\}$ and $[T]$ respectively. $x_{i, j, s}$ indicates whether job $j$ is processed on machine $i$ with starting time $s$.  The objective to minimize is the weighted completion time $\sum_{j}w_j\sum_{i, s}x_{i,j,s}(s+p_{i,j})$. Constraint~\eqref{LPC:RwC-job-j-scheduled} requires every job $j$ to be scheduled. Constraint~\eqref{LPC:RwC-congestion} says that on every machine $i$ at any time point $t$, only one job is being processed. Constraint~\eqref{LPC:RwC-job-length} says that if job $j$ is scheduled on $i$, then it can not be started after $T-p_{i,j}$.  Constraint~\eqref{LPC:RwC-non-negative} requires all variables to be nonnegative.

	\begin{theorem}
		\label{theorem:main-RwC}
		The LP relaxation \eqref{LP:RwC} for $R||\sum_j w_j C_j$ has an integrality gap of at most $1.5-c$, where $c = \frac{1}{6000}$.  Moreover, there is an algorithm that, given a valid fractional solution $x$ to \eqref{LP:RwC}, outputs a random valid schedule with expected cost at most $(1.5-c)\sum_j w_j\sum_{i,s}x_{i,j,s}(s+p_{i,j})$, in time polynomial in the number of non-zero variables of $x$.\footnote{We assume $x$ is given as a sequence of $(i, j, s, x_{i, j, s})$-tuples with non-zero $x_{i, j, s}$.}
	\end{theorem}
	
	The above algorithm leads to a $(1.5-c)$-approximation for $R||\sum_j w_j C_j$ immediately if $T$ is polynomially bounded. 
\ifdefined\CR
	The case when $T$ is super-polynomial will be handled in the full version of the paper. 
\else
	In Section~\ref{sec:superT}, we shall show how to handle the case when $T$ is super-polynomial. 
\fi

	\subsection{Our Techniques}
	
		A key technique in many of our results is the use of time-indexed LP relaxations. For the identical machine setting, we have variables $x_{j, t}$ indicating whether job $j$ is scheduled in the time-interval $(t - p_j, t]$; we can visualize $x_{j, t}$ as a rectangle of height $x_{j, t}$ with horizontal span $(t-p_j, t]$. With this visualization, it is straightforward to express the objective function, and formulate the machine-capacity constraints and the precedence constraints.   For the unrelated machine model, the LP we use is \eqref{LP:RwC}. (We used starting points to index intervals, as opposed to ending points; this is only for the simplicity of describing the algorithm.) Each $x_{i, j, s}$ can be viewed as a rectangle of height $x_{i, j, s}$ on machine $i$ with horizontal span $(s, s + p_{i, j}]$.
		The rectangle structures allow us to recover the previous state-of-art results, and furthermore to derive the improved approximation results by identifying the loose parts in these algorithms and analysis.
	
	
	\paragraph{$P|\tprec|\sum_j w_jC_j$} Let us first consider the scheduling problem on identical machines with job precedence constraints. The 4-approximation algorithm of Munier, Queyranne, and Schulz \cite{MQS98, QS06} used a convex programming that only contains the completion time variables $\{C_j\}_{j \in J}$.  After obtaining the vector $C$, we run the job-driven list scheduling algorithm, by considering jobs $j$ in increasing order of $C_j - p_j/2$.  To analyze the expected completion time of ${j^*}$ in the output schedule, focus on the schedule $\calS$ constructed by the algorithm at the time $j^*$ was inserted. Then, we consider the total length of busy and idle slots in $\calS$ before the completion of ${j^*}$ separately.  The length of busy slots can be bounded by $2C_{j^*}$, using the $m$-machine constraint. The length of idle slots can also be bounded by $2C_{j^*}$, by identifying a chain of jobs that resulted in the idle slots.  More generally, they showed that if jobs are considered in increasing order of $C_j - (1-\theta )p_j$ for $\theta \in [0, 1/2]$ in the list scheduling algorithm, the factor for idle slots can be improved to $1/(1-\theta)$ but the factor for busy slots will be increased to $1/\theta$. Thus, $\theta = 1/2$ gives the best trade-off. 
	
	The rectangle structure allows us to exam the tightness of the above factors more closely: though the $1/\theta$ factor for busy slots is tight for every individual $\theta \in [0, 1/2]$, it can not be tight for every such $\theta$.  Roughly speaking, the $1/\theta$ factor is tight for a job ${j^*}$ only when ${j^*}$ has small $p_{j^*}$, and all the other jobs $j$ considered before ${j^*}$ in the list scheduling algorithm has large $p_{j}$ and $C_{j} - \theta p_{j}$ is just smaller than $C_{j^*} - \theta p_{j^*}$.  However in this case, if we decrease $\theta$ slightly, these jobs $j$ will be considered after ${j^*}$ and thus the bound can not be tight for all $\theta \in [0, 1/2]$.  We show that even if we choose $\theta$ uniformly at random from $[0, 1/2]$, the factor for busy time slots remains $2$, as opposed to $\int_{\theta = 0}^{1/2} \frac{2}{\theta}\sfd\theta = \infty$. On the other hand, this decreases the factor for idle slots to $\int_{\theta = 0}^{1/2}\frac{2}{1-\theta}\sfd\theta = 2\ln2$, thus improving the approximation factor to $2+2\ln 2$. The idea of choosing a random point for each job $j$ and using them to decide the order in the list-scheduling algorithm has been studied before under the name ``$\alpha$-points'' 
	\cite{Goe97, HSS97, PSW98, CMN01}. The novelty of our result is the use of the rectangle structure to relate different $\theta$ values. 
	In contrast, solutions to the convex programming of \cite{MQS98} and the weak time-indexed LP relaxation of \cite{SS02} lack such a structure.  
	
	\paragraph{$P|\tprec, p_j=1|\sum_j w_jC_j$} When jobs have uniform length, the approximation ratio of the algorithm of \cite{MQS98} improves to $3$. In this case, the $\theta$ parameter in the above algorithm becomes useless since all jobs have the same length.  Taking the advantage of the uniform job length, the factor for idle time slots improves 1, while the factor for busy slots remains 2. This gives an approximation factor of $3$ for the special case. 
	
	To improve the factor of $3$, we use another randomized procedure to decide the order of jobs in the list scheduling algorithm.  For every $\theta \in [0, 1]$, let $M^\theta_j$ be the first time when we scheduled $\theta$ fraction of job $j$ in the fractional solution. Then we randomly choose $\theta \in [0, 1]$ and consider jobs the increasing order of $M^\theta_j$ in the list-scheduling algorithm.  This algorithm can recover the factor of 1 for total length of idle slots and 2 for total length of busy slots.
	
	We again use the rectangle structure to discover the loose part in the analysis. With uniform job size, the idle slots before a job $j$ are caused only by the precedence constraints: if the total length of idle slots before the completion time of $j$ is $a$, then there is a precedence-chain of $a$ jobs ending at $j$; in other words, $j$ is at depth at least $a$ in the precedence graph. In order for the factor 1 for idle slots to be tight, we need to have $a \approx C_j$.  We show that if this happens, the factor for busy time slots shall be much better than $2$. Roughly speaking, the factor of 2 for busy time slots is tight only if $j$ is scheduled evenly among $[0, 2C_j]$.  However, if $j$ is at depth-$a$ in the dependence graph, it can not be scheduled before time $a \approx C_j$ with any positive fraction.  A quantification of this argument allows us to derive the improved approximation ratio $1+\sqrt{2}$ for this special case.
	
	\paragraph{$Q|\tprec|\sum_j w_j C_j$} Our $O(\log m/\log \log m)$-approximation for related machine scheduling is a simple one. As mentioned earlier, by losing a constant factor in the approximation ratio, we can convert the problem of minimizing the weighted completion time to that of minimizing the makespan, i.e, the problem $Q|\tprec|C_{\max}$.
	To minimize the makespan, the algorithm of Chudak and Shmoys \cite{CS97} partitions machines into $O(\log m)$ groups according to their speeds.  Based on their LP solution, we assign each job $j$ to a group of machines. Then we run the machine-driven list-scheduling algorithm, subject to the precedence constraint, and the constraint that each job can only be scheduled to a machine in its assigned group. The final approximation ratio is the sum of two factors: one from grouping machines with different speeds into the same group, which is $O(1)$ in \cite{CS97}, and the other from the number of different groups, which is $O(\log m)$ in \cite{CS97}. To improve the ratio, we make the speed difference between machines in the same group as large as $\Theta(\log m/\log\log m)$, so that we only have $O(\log m/\log \log m)$ groups. Then, both factors become $O(\log m/\log\log m)$, leading to an $O(\log m/\log\log m)$-approximation for the problem.  One remark is that in the algorithm of \cite{CS97},  the machines in the same group can be assumed to have the same speed, since their original speeds only differ by a factor of 2. 
	In our algorithm, we have to keep the original speeds of machines, to avoid a multiplication of the two factors in the approximation ratio.
	
		
	\paragraph{$R||\sum_j w_jC_j$} Then we sketch how we use our time-indexed LP to recover the $(1.5-c)$-approximation of \cite{BSS16} (with much better constant $c$), for the scheduling problem on unrelated machines to minimize total weighted completion time, namely $R||\sum_{j} w_j C_j$. 
	
	The dependence rounding procedure of \cite{BSS16} is the key component leading to a better than 1.5 approximation for $R||\sum_{j} w_j C_j$. It takes as input a grouping scheme: for each machine $i$, the jobs are partitioned into groups with total fractional assignment on $i$ being at most 1. The jobs in the same group for $i$ will have strong negative correlation towards being assigned to $i$.  To apply the theorem, they first solve the lift-and-project SDP relaxation for the problem,  and construct a grouping scheme based on the optimum solution to the SDP relaxation.  For each machine $i$, the grouping algorithm will put jobs with similar Smith-ratios in the same group, as the 1.5-approximation ratio is caused by conflicts between these jobs.  With the strong negative correlation, the approximation ratio can be improved to $(1.5 - c)$ for a tiny constant $c = 1/(108\times 20000)$.	
		
	We show that the natural time-indexed relaxation \eqref{LP:RwC} for the problem suffices to give a $(1.5-c)$-approximation. 
	To apply the dependence rounding procedure, we need to construct a grouping for every machine $i$.  In our recovered 1.5-approximation algorithm for the problem using \eqref{LP:RwC}, the expected completion time of $j$ is at most $\sum_{i,s}x_{i,j, s}(s+1.5p_{i,j})$, i.e, the average starting time of $j$ plus 1.5 times the average length of $j$ in the LP solution. This suggests that a job $j$ is bad only when its average starting time is very small compared to its average length in the LP solution.  Thus, for each machine $i$, the bad jobs are those with a large weight of scheduling intervals near the beginning of the time horizon.  If these bad intervals for two bad jobs $j$ and $j'$ have large overlap, then they are likely to be put into the same group for $i$.  To achieve this, we construct a set of disjoint basic blocks $\{(2^a, 2^{a+1}]: a \geq -2\}$ in the time horizon. A bad job will be assigned to a random basic block contained in its scheduling interval and two bad jobs assigned to the same basic block will likely to be grouped together.  Besides the improved approximation ratio, we believe the use of \eqref{LP:RwC} will shed light on getting an approximation ratio for the problem that is considerably better than 1.5, as it is simpler than the lift-and-project SDP of \cite{BSS16}.  Another useful property of our algorithm is that the rounding procedure is oblivious to the weights of the jobs; this may be useful when we consider some variants of the problem. 
	\medskip
		
	Finally,  we remark that Theorems~\ref{theorem:main-PprecwC},~\ref{theorem:main-PprecwCuniform} and~\ref{theorem:main-QprecwC} can be easily extended to handle job arrival times.  However, to deliver the key ideas more efficiently, we chose not to consider arrival times. 
	
	\subsection{Other Related Work} There is a vast literature on approximating algorithms for scheduling problems to minimize the total weighted completion time. Here we only discuss the ones that are most relevant to our results; we refer readers to \cite{CK04} for a more comprehensive overview.   When there are no precedence constraints, the problems of minimizing total weighted completion time on identical and related machines ($P||\sum_j w_jC_j$ and $Q||\sum_j w_jC_j$) admit PTASes (\cite{SW99c, CK01}). 
	For the problem of scheduling jobs on unrelated machines with job arrival times to minimize weighted completion time ($R|r_j|\sum_j w_j C_j$), many classic results give $2$-approximation algorithms (\cite{Sku01, SS02, KMP08}); recently Im and Li  \cite{IL16} gave a 1.8687-approximation for the problem, solving a long-standing open problem.  Skutella \cite{Sku16} gave a $\sqrt{e}/(\sqrt{e}-1)\approx 2.542$-approximation algorithm for the single-machine scheduling problem with precedence constraints and job release times, improving upon the previous $e \approx 2.718$-approximation \cite{SS97a}.
	
	Makespan is an objective closely related to weighted completion time.  As we mentioned, for $P|\tprec|C_{\max}$, the Graham's list scheduling algorithm gives a $2$-approximation, which is the best possible under a stronger version of UGC \cite{BK09, Sve10}. For the special case of the problem $Pm|\tprec, p_j = 1|C_{\max}$ where there are constant number of machines and all jobs have unit size, the recent breakthrough result of Levey and Rothvoss \cite{LR16} gave a $(1+\epsilon)$-approximation with running time $\exp\left(\exp\left(O_{m, \epsilon}(\log^2\log n)\right)\right)$, via the LP hierarchy of the natural LP relaxation for the problem. On the negative side, it is not even known whether $Pm|\tprec, p_j = 1|C_{\max}$ is NP-hard or not. For the problem $R||C_{\max}$, i.e, the scheduling of jobs on unrelated machines to minimize the makespan, the classic result of Lenstra, Shmoys and Tardos \cite{LST90} gives a $2$-approximation, which remains the best algorithm for the problem.  Some efforts have been put on a special case of the problem, where each job $j$ has a size $p_j$ and $p_{i, j} \in \{p_j, \infty\}$ for every $i \in M$ (the model is called restricted assignment model.) \cite{EKS08, Sve11, CKL15, KL17}.   
	
\ifdefined\CR
\else
	When jobs have arrival times, the flow time of a job, which is its completion time minus its arrival time, is a more suitable measurement of quality of service. There is a vast literature on scheduling algorithms with flow time related objectives \cite{LR97, CK02, GK07, Sit09, BP10, BK15}; since they are much harder to approximate than completion time related objectives,  most of these works can not handle precedence constraints and need to allow preemptions of jobs.
\fi

	\paragraph{Organization} 
\ifdefined\CR
	The proofs of Theorems~\ref{theorem:main-PprecwC} to~\ref{theorem:main-QprecwC} are given in Sections~\ref{sec:PwC} to~\ref{sec:QwC} respectively.   Due to the space limit, the proof of Theorem~\ref{theorem:main-RwC} is deferred to the full version of the paper. 
\else
	The proofs of Theorems~\ref{theorem:main-PprecwC} to~\ref{theorem:main-RwC} are given in Sections~\ref{sec:PwC} to~\ref{sec:RwC} respectively.  
\fi
	Throughout this paper, we assume the weights, lengths of jobs are integers. Let $T$ be the maximum makespan of any ``reasonable'' schedule.  For problems \PprecwC and $R||\sum_jw_jC_j$, we {\ifdefined \CR \else first \fi} assume $T$ is polynomial in $n$. By losing a $1+\epsilon$ factor in the approximation ratio, we can handle the case where $T$ is super-polynomial. 
\ifdefined \CR
	This and the other omitted proofs can be found in the full version of the paper.
\else
	This is shown in Section~\ref{sec:superT}.
\fi
			
%
%
%
%
%
%
%
%
%
%



\section{Scheduling on Identical Machines with Job Precedence Constraints}
\label{sec:PwC}

	In this section we give our $(2 + 2\ln 2 + \epsilon)$-approximation for the problem of scheduling precedence-constrained jobs on identical machines, namely $P|\tprec|\sum_{j} w_j C_j$.  We solve \eqref{LP:PprecWc} and run the job-driven list-scheduling algorithm of \cite{MQS98} with a random order of jobs. 


	\subsection{Time-Indexed LP Relaxation for $P|\tprec|\sum_j w_jC_j$} In the identical machine setting,  we do not need to specify which machine each job is assigned to; it suffices to specify a scheduling interval $(t-p_j, t]$ for every job $j$. A folklore result says that a set of intervals can be scheduled on $m$ machines if and only if their \emph{congestion} is at most $m$: i.e, the number of intervals covering any time point is at most $m$. Given such a set of intervals, there is a simple greedy algorithm to produce the assignment of intervals to machines.  Thus, in our LP relaxation and in the list-scheduling algorithm, we  focus on finding a set of intervals with congestion at most $m$.
	
	We use \eqref{LP:PprecWc} for both \PprecwC and \PprecupwC. Let $T = \sum_j p_j$ be a trivial upper bound on the makespan of any reasonable schedule. In the LP relaxation, we have a variable $x_{j, t}$ indicating whether job $j$ is scheduled in $(t-p_j, t]$, for every $j \in J$ and $t \in [T]$. Throughout this and the next section, $t$ and $t'$ are restricted to be integers in $[T]$, and $j$, $j'$ and $j^*$ are restricted to be jobs in $J$.
	
\ifdefined\CR
	\begin{equation}
	\min \qquad \sum_{j}w_j\sum_{t}x_{j, t}t \quad \text{s.t.} \tag{$\text{LP}_{\text{P}|\text{prec}|\text{wC}}$} \label{LP:PprecWc}
	\end{equation}
	\vspace*{-28pt}
	
		\begin{align}
			\sum_t x_{j, t} &= 1 & \quad \forall& j \label{LPC:PwC-job-scheduled} \\[1pt]
			\sum_{j, t \in [t', t'+p_j)} x_{j,t} &\leq m & \quad \forall& t' \label{LPC:PwC-congestion} \\
			\sum_{t < t' +p_{j'}} x_{j', t} &\leq \sum_{t < t'}x_{j,t} &\quad  \forall & j, j',t': j \prec j' \label{LPC:PwC-precedence}\\
			x_{j, t} & = 0 & \quad \forall &j, t < p_j\label{LPC:PwC-length} \\
			x_{j, t} &\geq 0 & \quad \forall& j, t \label{LPC:PwC-non-negative}
		\end{align}
\else
	\begin{equation}
		\min \qquad \sum_{j}w_j\sum_{t}x_{j, t}t \qquad \qquad \text{s.t.} \tag{$\text{LP}_{\text{P}|\text{prec}|\text{wC}}$} \label{LP:PprecWc}
	\end{equation}
	\vspace*{-28pt}
	
	\noindent
	\begin{minipage}{0.37\textwidth}
		\vspace*{8pt}
		\begin{align}
			\sum_t x_{j, t} &= 1 & \quad \forall& j \label{LPC:PwC-job-scheduled} \\[1pt]
			\sum_{j, t \in [t', t'+p_j)} x_{j,t} &\leq m & \quad \forall& t' \label{LPC:PwC-congestion} \\
			\nonumber
		\end{align}
	\end{minipage}\hspace*{0.06\textwidth}
	\begin{minipage}{0.56\textwidth}
		\begin{align}
			\sum_{t < t' +p_{j'}} x_{j', t} &\leq \sum_{t < t'}x_{j,t} &\quad  \forall & j, j',t': j \prec j' \label{LPC:PwC-precedence}\\
			x_{j, t} & = 0 & \quad \forall &j, t < p_j\label{LPC:PwC-length} \\
			x_{j, t} &\geq 0 & \quad \forall& j, t \label{LPC:PwC-non-negative}
		\end{align}
	\end{minipage}
\fi
	
	The objective function is $\sum_{j}w_j \sum_t x_{j, t}t$, i.e, the total weighted completion time over all jobs. Constraint~\eqref{LPC:PwC-job-scheduled} requires every job $j$ to be scheduled. Constraint~\eqref{LPC:PwC-congestion} requires that at every time point $t'$, at most $m$ jobs are being processed. 
	Constraint~\eqref{LPC:PwC-precedence} requires that for every $j \prec j'$ and $t'$, $j'$ completes before $t' + p_{j'}$ only if $j$ completes before time $t'$. A job $j$ can not complete before $p_j$ (Constraint~\eqref{LPC:PwC-length}) and all variables are non-negative (Constraint~\eqref{LPC:PwC-non-negative}).  

	We solve \eqref{LP:PprecWc} to obtain $x \in [0, 1]^{J \times [T]}$.  Let $C_j = \sum_{t}x_{j, t}t$ be the completion time of $j$ in the LP solution. Thus, the value of the LP is $\sum_{j} w_j C_j$. For every $\theta \in [0, 1/2]$, we define $M^\theta_j = C_j - (1-\theta)p_j$. Our algorithm is simply the following:  choose $\theta$ uniformly at random from $(0, 1/2]$, and output the schedule returned by job-driven-list-scheduling($M^\theta$) (described in Algorithm~\ref{alg:PwC-list-scheduling}).

	\begin{algorithm}
		\caption{job-driven-list-scheduling$\left(M\right)$} \label{alg:PwC-list-scheduling}
		\textbf{Input}:  a vector $M \in \R_{\geq 0}^J$ used to decide the order of scheduling, s.t. if $j \prec j'$, then $M_j < M_{j'}$  \\
		\textbf{Output}: starting and completion time vectors ${\tilde S}, \tilde C \in \R_{\geq 0}^J$
		\begin{algorithmic}[1]
			\STATE \bffor every $j \in J$ in non-decreasing order of $M_j$, breaking ties arbitrarily
			\STATE \hspace*{\algorithmicindent} let $t \gets \max_{j' \prec j}\tilde C_{j'}$, or $t \gets 0$ if $\{j' \prec j\} = \emptyset$
			\STATE \hspace*{\algorithmicindent} find the minimum $s \geq t$ such that we can schedule $j$ in interval $(s, s + p_j]$, without increasing the congestion of the schedule to $m+1$
			\STATE \hspace*{\algorithmicindent} ${\tilde S}_j \gets s, \tilde C_j \gets s + p_j$, and schedule $j$ in $({\tilde S}_j, \tilde C_j]$
			\STATE \bfreturn $({\tilde S}, \tilde C)$
		\end{algorithmic}
	\end{algorithm}

	We first make a simple observation regarding the $C$ vector, which follows from the constraints in the LP.
	\begin{claim}
		\label{claim:difference-for-C}
		For every pair of jobs $j, j'$ such that $j \prec j'$, we have $C_j + p_{j'} \leq C_{j'}$.
	\end{claim}
\ifdefined\CR
\else
	\vspace*{-30pt}
	\begin{flalign*}
	Proof. && & \quad C_j + p_{j'} = \sum_{t'}x_{j,t'} t'  +p_{j'} =  \sum_{t',t\leq t'} x_{j, t'} + p_{j'}= \sum_{t}\left(1-\sum_{t'< t}x_{j,t'}\right) +p_{j'}\\
	&& &\leq \sum_t\left(1-\sum_{t' <  t + p_{j'}}x_{j', t'}\right) + p_{j'} = \sum_{t,t'\geq t+p_{j'}}x_{j',t'} + p_{j'} = \sum_{t'} x_{{j'}, t'} \left|\set{t:t \leq t' - p_{j'}}\right| +p_{j'} \\
	&&  &= \sum_{t' \geq p_{j'}}(t'-p_{j'})x_{j',t'} + p_{j'} = \sum_{t'}t'x_{j',t'} - p_{j'} + p_{j'}= C_{j'}. &
	\end{flalign*}
	The inequality used Constraint~\eqref{LPC:PwC-precedence}; some of the equalities used Constraint~\eqref{LPC:PwC-job-scheduled} and~\eqref{LPC:PwC-length} and the definitions of $C_j$ and $C_{j'}$. \hfill \qed 
\fi	
	
	Indeed, our analysis does not use the full power of Constraint~\eqref{LPC:PwC-precedence}, except for the above claim which is implied by the constraint.  Thus, we could simply use $C_j + p_{j'} \leq C_{j'}$ (along with the definitions of $C_j$'s) to replace Constraint~\eqref{LPC:PwC-precedence} in the LP. However, in the algorithm for the problem with unit job lengths (described in Section~\ref{sec:PprecwCUniform}), we do need Constraint~\eqref{LPC:PwC-precedence}. To have a unified LP for both problems, we chose to use Constraint~\eqref{LPC:PwC-precedence}.  Our algorithm does not use $x$-variables, but we need them in the analysis.
	

	\subsection{Analysis}  
	Our analysis is very similar to that in \cite{MQS98}.  	We fix a job $j^*$ from now on and we shall upper bound $\frac{\E[\tilde C_{j^*}]}{C_{j^*}}$.   Notice that once $j^*$ is scheduled by the algorithm, $\tilde C_{j^*}$ is determined and will not be changed later.  Thus, we call the schedule at the moment the algorithm just scheduled $j^*$ the \emph{final schedule}. 	
	
	We can then define idle and busy points and slots w.r.t this final schedule. We say a time point $\tau \in (0, T]$ is \emph{busy} if the congestion of the intervals at $\tau$ is $m$ in the schedule (in other words, all the $m$ machines are being used at $\tau$ in the schedule); we say $\tau$ is \emph{idle} otherwise.  We say a left-open-right-closed interval (or slot) $(\tau, \tau']$ (it is possible that $\tau = \tau'$, in which case the interval is empty) is idle (busy, resp.) if all time points in  $(\tau, \tau']$ are idle (busy, resp.). 
	
	Then we analyze the total length of busy and idle time slots before $\tilde C_{j^*}$ respectively, w.r.t the final schedule.  For a specific $\theta \in (0, 1/2]$, the techniques in \cite{MQS98} can bound the total length of idle slots by $\frac{C_{j^*}}{1-\theta}$ and the total length of busy slots by $\frac{C_{j^*}}{\theta}$. Thus choosing $\theta = 1/2$ gives the best $4$-approximation, which is the best using this analysis.  Our improvement comes from the bound on the total length of busy time slots. We show that the expected length of busy slots before $\tilde C_{j^*}$ is at most $2 C_{j^*}$, which is much better than the bound $\E_{\theta\sim_R (0, 1/2]}\frac{C_{j^*}}{\theta} = \infty$ given by directly applying the bound for every $\theta$. We remark that the $\frac{C_{j^*}}{\theta}$ bound for each individual $\theta$ is tight and thus can not be improved; our improvement comes from considering all possible $\theta$'s together.


	\paragraph{Bounding the Expected Length of Idle Slots} 
		We first bound the total length of idle slots before $\tilde C_{j^*}$, the completion time of job $j^*$ in the schedule produced by the algorithm. Lemma~\ref{lemma:PwC-previous-j} and~\ref{lemma:PwC-bound-idle} are established in \cite{MQS98} and 
\ifdefined\CR
		their proofs can be found in the full version of the paper for completeness.
\else
		we include their proofs for completeness. 
\fi
		\begin{lemma}
			\label{lemma:PwC-previous-j}  Let $j \in J$ be a job in the final schedule with ${\tilde S}_j > 0$. Then we can find a job $j'$ such that 
			\begin{itemize}[topsep=3pt,itemsep=0pt]
				\item either $j' \prec j$ and $(\tilde C_{j'}, {\tilde S}_j]$ is busy,
				\item or $M_{j'} \leq M_j, S_{j'} < S_{j}$ and $(S_{j'}, {\tilde S}_j]$ is busy.
			\end{itemize}
		\end{lemma}

\ifdefined\CR
\else	
	\begin{proof}
		Recall that busy and idle slots are defined w.r.t the final schedule, i.e, the schedule at the moment we just scheduled $j^*$.  We first assume $({\tilde S}_j -1, {\tilde S}_j]$ is idle. Then before the iteration for $j$,  the interval $({\tilde S}_j - 1, {\tilde S}_j +p_j]$ is available for scheduling and thus scheduling $j$ at $({\tilde S}_j - 1, {\tilde S}_j - 1 + p_j]$ will not violate the congestion constraint. Therefore, it must be the case that ${\tilde S}_j = \max_{j' \prec j}\tilde C_{j'}$. So, there is a job $j' \prec j$ such that $\tilde C_{j'} = {\tilde S}_j$. Since $(\tilde C_{j'}, {\tilde S}_j] = \emptyset$, the first property holds.
		
		So we can assume $({\tilde S}_j -1, {\tilde S}_j]$ is busy.  Let $(\tau, \tau']$ be the maximal busy slot that contains $({\tilde S}_j - 1, {\tilde S}_j]$. If there is a job $j' \prec j$ such that $\tilde C_{j'} \geq \tau$, then the first property holds since $(\tilde C_{j'}, {\tilde S}_j]$ is busy.  So, we can assume that no such job $j'$ exists. 
		
		As $\tau$ is the starting point of a maximal busy slot,  there is a job $j_1$ with $S_{j_1} = \tau < {\tilde S}_j$. If $\tilde C_{j_1} < {\tilde S}_j$, then there is a job $j_2$ with $\tilde C_{j_1} = S_{j_2} < {\tilde S}_j$, as $(C_{j_1}, C_{j_1} + 1]$ is busy.  If $\tilde C_{j_2} < {\tilde S}_j$ then there is job $j_3$ with $\tilde C_{j_2} = \tilde S_{j_3} < {\tilde S}_j$. We can repeat this process to find a job $j_k$ such that $\tau \leq S_{j_k} < {\tilde S}_j \leq \tilde C_{j_k}$.  Let $j' = j_k$; we claim that $M_{j'} \leq M_j$. Otherwise, $j$ is considered before $j'$ in the list scheduling algorithm. At the iteration for $j$, the interval $(\tilde S_{j'}, \tilde C_j]$ is available. Since we assumed that there are no jobs $j'' \prec j$ such that $\tilde C_{j''} \geq \tau$, $j$ should be started at $S_{j'}$,  a contradiction.  Thus, $M_{j'} \leq M_j$ and $j'$ satisfies the second property of the lemma. 
	\end{proof}
\fi	
	Applying Lemma~\ref{lemma:PwC-previous-j} repeatedly, we can identify a chain of jobs whose scheduling intervals cover all the idle slots before $\tilde C_{j^*}$, which can be used to bound the total length of these slots. This leads to the following lemma from \cite{MQS98}:
		
	\begin{lemma}
		\label{lemma:PwC-bound-idle}
		The total length of idle time slots before $\tilde C_{j^*}$ is at most  $\frac{C_{j^*}}{1-\theta}$.
	\end{lemma}

\ifdefined\CR
\else
	\begin{proof}
		Let  $j_0 = j^*$; if $S_{j_0} > 0$, we apply Lemma~\ref{lemma:PwC-previous-j} for $j = j_0$ to find a job $j'$ and let $j_1 = j'$. If $S_{j_1} > 0$, then we apply the lemma again for $j = j_1$ to find a job $j'$ and let $j_2 = j'$. We can repeat the process until we reach a job $j_k$ with $S_{j_k} = 0$. For convenience, we shall \emph{revert the sequence} so that $j^* = j_k$ and $S_{j_0} = 0$.  Thus, we have found a sequence $j_0, j_1, \cdots, j_k = j^*$ of jobs such that $S_{j_0} = 0$, and for each $\ell = 0, 1, 2, \cdots, k-1$, either (i) $j_{\ell} \prec j_{\ell+1}$ and $\left(\tilde C_{j_{\ell}}, S_{j_{\ell+1}}\right]$ is busy, or (ii) $M_{j_{\ell}} \leq M_{j_{\ell+1}}$, $S_{j_{\ell}} < S_{j_{\ell+1}}$ and $(S_{j_{\ell}}, S_{j_{\ell+1}}]$ is busy. 
		
		We say $\ell$ is of type-1 if $\ell$ satisfies (i); otherwise, we say $\ell$ is of type-2 ($\ell$ must satisfy (ii)). The interval $(0, S_{j^*}]$ can be broken into $k$ intervals: $(S_{j_0}, S_{j_1}], (S_{j_1}, S_{j_2}], \cdots, (S_{j_{k-1}}, S_{j_k}]$.  If some $\ell$ is of type-2, then $(S_{j_\ell}, S_{j_{\ell+1}}]$ is busy; if $\ell$ is of type-1, then $(\tilde C_{j_\ell}, S_{j_{\ell+1}}]$ is busy.  Let $L$ be the set of type-1 indices $\ell \in [0, k-1]$. Then, all the idle slots in $(0, S_{j^*}]$ are contained in $\displaystyle \union_{\ell \in L}(S_{j_\ell}, \tilde C_{j_\ell}]$.  With this observation, we can bound the total length of idle slots in $(0, S_{j^*}]$ by
		\begin{align*}
			&\  \sum_{\ell \in L} p_{j_\ell}\leq \sum_{\ell \in L} \frac{1}{1-\theta} (M^\theta_{j_{\ell+1}} - M^\theta_{j_{\ell}})
			 \leq \frac{1}{1-\theta}\sum_{\ell = 0}^{k-1} (M^\theta_{j_{\ell+1}} - M^\theta_{j_{\ell}}) \leq \frac{M^\theta_{j^*}}{1-\theta} = \frac{C_{j^*}}{1-\theta} - p_{j^*}.
		\end{align*}
		The first inequality holds since $p_{j_\ell } = \frac{1}{1-\theta}\left(C_{j_\ell} - M^\theta_{j_\ell}\right) \leq \frac{1}{1-\theta}\left(C_{j_{\ell+1}} - p_{j_{\ell+1}} - M^\theta_{j_\ell}\right) \leq \frac{1}{1-\theta}\big(M^\theta_{j_{\ell+1}} - M^\theta_{j_\ell}\big)$, due to Claim~\ref{claim:difference-for-C}.  The second inequality is by the fact that $M_{j_\ell} \leq M_{j_{\ell+1}}$ for every $\ell \in [0, k-1]$ and the third inequality is by $j_k = j^*$ and $M^\theta_{j_0} \geq 0$. The equality is by the definition of $M^\theta_{j^*}$. Thus, the total length of idle slots in $(0, \tilde C_{j^*}]$ is at most $\frac{C_{j^*}}{1-\theta}$.
	\end{proof}
\fi	

	Thus, the expected length of idle slots before $\tilde C_{j^*}$, over all choices of $\theta$, is at most
	\begin{align}
		\int_{\theta = 0}^{1/2} \frac{C_{j^*}}{1-\theta} 2\sfd\theta = \left(2\ln\frac{1}{1-\theta}\Big|_{\theta = 0}^{1/2}\right)C_{j^*} = (2\ln 2)C_{j^*}. \label{inequ:PwC-bound-idle-integral}
	\end{align}

	\paragraph{Bounding the Expected Length of Busy Slots}
	We now proceed to bound the total length of busy slots before $\tilde C_{j^*}$.  This is the key to our improved approximation ratio. For every $\theta \in [0, 1/2]$, let  $J_\theta = \{j: M^\theta_j \leq \tilde C_{j^*}\}$. Thus, if $\theta < \theta'$, we have $J_\theta \supseteq J_{\theta'}$. For every $\theta \in [0, 1/2]$ and $j \in J_0$, define $\theta_j = \sup\set{\theta \in [0,1/2]: j \in J_\theta}$; this is well-defined since $j \in J_0$.  For any subset $J' \subseteq J$ of jobs, we define $p(J') = \sum_{j \in J'} p_j$ to be the total length of all jobs in $J'$. 
	\begin{lemma}
		\label{lemma:PwC-bound-busy}
		For a fixed $\theta \in (0, 1/2]$, the total length of busy slots before $\tilde C_{j^*}$ is at most $\frac{1}{m}p(J_\theta)$.
	\end{lemma}
	\begin{proof}
		The total length of busy time slots in $(0, \tilde C_{j^*}]$ is at most  $\frac1m$ times the total length of jobs scheduled so far, which is at most \begin{flalign*}
			&& \frac1m\sum_{j \in J: M^\theta_j \leq M^\theta_{j^*}}p_j \leq \frac1m\sum_{j \in J: M^\theta_j \leq  C_{j^*}}p_j = \frac1m p(J_\theta).&& \qedhere
		\end{flalign*}
	\end{proof}

	The key lemma for our improved approximation ratio is an upper bound on the above quantity when $\theta$ is uniformly selected from  $(0, 1/2]$:	
	\begin{lemma}
		\label{lemma:PwC-bound-busy-integral}
		$\displaystyle \int_{\theta=0}^{1/2}p(J_\theta)\sfd \theta \leq m  C_{j^*}$.
	\end{lemma}
	
	\begin{proof} 
		Notice that we have
\ifdefined\CR
		\begin{align*}
		&\qquad \int_{\theta=0}^{1/2}p(J_\theta)\sfd \theta
		= \int_{\theta = 0}^{1/2} \sum_{j \in J_{0}} p_j \mathbf 1_{j \in J_\theta} \sfd \theta  \\
		&= \sum_{j \in J_{0}} p_j \int_{\theta = 0}^{1/2}\mathbf 1_{j \in J_\theta} \sfd \theta  
		= \sum_{j \in J_{0}} \theta_j p_j.
		\end{align*}
\else
		\begin{align*}
			\int_{\theta=0}^{1/2}p(J_\theta)\sfd \theta
			= \int_{\theta = 0}^{1/2} \sum_{j \in J_{0}} p_j \mathbf 1_{j \in J_\theta} \sfd \theta  
			= \sum_{j \in J_{0}} p_j \int_{\theta = 0}^{1/2}\mathbf 1_{j \in J_\theta} \sfd \theta  
			= \sum_{j \in J_{0}} \theta_j p_j.
		\end{align*}
\fi
		Thus, it suffices to prove that $\sum_{j \in J_0} \theta_j p_j \leq m C_{j^*}$. To achieve this, we construct a set of axis-parallel rectangles. For each $j \in J_0$ and $t$ such that $x_{j, t} > 0$, we place a rectangle with height $x_{j,t}$ and horizontal span $(t-p_j, t-p_j+2\theta_j p_j]$.  The total area of all the rectangles for $j$ is exactly $2\theta_jp_j$.  Notice that $\theta_j \leq 1/2$ and thus $(t - p_j, t - p_j + 2\theta_j p_j] \subseteq (t-p_j, t]$. 
		
		Notice that $\sum_{t}x_{j, t}(t - p_j + \theta_jp_j) =  C_j - (1-\theta_j)p_j = M^{\theta_j}_{j} \leq  C_{j^*}, \sum_{t}x_{j,t} = 1$, and $t-p_j + \theta_j p_j$ is the mass center\footnote{Here, we use mass center for the horizontal coordinate of the mass center, since we are not concerned with the vertical positions of rectangles.} of the rectangle for $(j, t)$. Thus, the mass center of the union of all rectangles for $j$ is at most $ C_{j^*}$. This in turn implies that the mass center of the union of all rectangles over all $j \in J_0$ and $t$, is at most $ C_{j^*}$. Notice that for every $t \in (0, T]$, the total height of all rectangles covering $t$ is at most $m$, by Constraint~\eqref{LPC:PwC-congestion}, and the fact that $(t-p_j, t-p_j + 2\theta_j p_j] \subseteq (t-p_j, t]$ for every $j \in J_0$ and $t$. Therefore, the total area of rectangles for all $j \in J_0$ and $t$ is at most $2m C_{j^*}$ (otherwise, the mass center will be larger than $C_{j^*}$).  So, we have $\sum_{j \in J_0}2\theta_jp_j \leq 2m C_{j^*}$, which finishes the proof of the lemma. 
	\end{proof}
	
	Thus, by Lemma~\ref{lemma:PwC-bound-busy} and Lemma~\ref{lemma:PwC-bound-busy-integral}, the expected length of busy time slots before $\tilde C_{j^*}$ is at most 
	\begin{align}
		\label{inequ:PwC-bound-busy-integral}
		\int_{\theta=0}^{1/2}\frac{p(J_\theta)}{m}2\sfd \theta = \frac{2}{m} \int_{ \theta = 0}^{1/2}p(J_\theta)\sfd \theta \leq 2C_{j^*}.
	\end{align}

	Thus, by Inequalities~\eqref{inequ:PwC-bound-idle-integral} and \eqref{inequ:PwC-bound-busy-integral}, we have
	\begin{align*}
		\E\left[\tilde C_{j^*}\right] \leq (2\ln2)C_{j^*} + 2 C_{j^*}  = (2+ 2\ln 2) C_{j^*}.
	\end{align*}
	Thus, we have proved the $2+2\ln2 + \epsilon \leq (3.387 + \epsilon)$-approximation ratio for our algorithm, finishing the proof of Theorem~\ref{theorem:main-PprecwC}.
	
	\paragraph{Remarks} One might wonder if choosing a random $\theta$ from $[0, \theta^*]$ for a different $\theta^*$ can improve the approximation ratio.  For $\theta^* \leq 1/2$, the ratio we can obtain is $\frac{1}{\theta^*} + \frac{1}{\theta^*}\int_{\theta = 0}^{\theta^*}\frac{1}{1-\theta}\sfd\theta = \frac{1}{\theta^*}  + \frac{1}{\theta^*}\ln\frac{1}{\theta^*}$; that is, the first factor (for busy time slots) will be increased to $1/\theta^*$ and the second factor (for idle time slots) will be decreased to $\frac{1}{\theta^*}\ln\frac{1}{\theta^*}$. This ratio is minimized when $\theta^* = 1/2$.  If $\theta^* > 1/2$, however, the first factor does not improve to $1/\theta^*$, as the proof of Lemma~\ref{lemma:PwC-bound-busy-integral} used the fact that $\theta_j \leq 1/2$ for each $j$. Thus, using our analysis, the best ratio we can get is $2+2\ln 2$. 
		

	\section{Scheduling Unit-Length Jobs on Identical Machines with Job Precedence Constraints}
	\label{sec:PprecwCUniform}
	In this section,  we give our $(1+\sqrt{2})$-approximation algorithm for \PprecupwC. Again, we solve \eqref{LP:PprecWc} to obtain $x$; define $C_j = \sum_{t} x_{j, t}$ for every $j \in J$. For this special case, we define the random $M$-vector differently. In particular, it depends on the values of $x$ variables. For every $j \in J$ and $\theta \in (0, 1]$, define $M^\theta_j$ to be the minimum $t$ such that $\sum_{t'  = 1}^{t} x_{j, t'} \geq \theta$.  Notice that $C_j = \int_{\theta = 0}^{1} M^\theta_j\sfd \theta$. Our algorithm for \PprecupwC chooses $\theta$ uniformly at random from $(0, 1]$, and then call  job-driven-list-scheduling$(M^\theta)$ and output the returned schedule.
	
	For every $j \in J$, we define  $a_j$ to be the largest $a$ such that there exists a sequence of $a$ jobs $j_1 \prec j_2 \prec j_3 \prec \cdots \prec j_a = j$. Thus, $a_j$ is the ``depth'' of $j$ in the precedence graph. 
	\begin{claim}
		\label{claim:chain}
		For every $j \in J$ and $t < a_j$, we have $x_{j, t} = 0$.
	\end{claim}
\ifdefined\CR
\else
	\begin{proof}
		By the definition of $a_j$, there is a sequence of $a  := a_j$ jobs $j_1 \prec j_2 \prec \cdots \prec j_a$ such that $j_a = j$. Let $t_a = t$.  If $x_{j_a, t_a} > 0$, then by Constraint~\eqref{LPC:PwC-precedence}, there is an $t_{a-1} \leq t_a-1$ such that $x_{j_{a-1}, t_{a-1}} > 0$.  Then, there is an $t_{a-2} \leq t_{a-1}-1 \leq t_a-2$ such that $x_{j_{a-2}, t_{a-2}} > 0$. Repeating this process, there will be an $t_1 \leq t_a - (a-1)$ such that $x_{j_1, t_1} > 0$. This contradicts the fact that $t_a= t \leq a-1$. 
	\end{proof}
\fi	
	Again, we fix a job $j^*$ and focus on the schedule at the moment the algorithm just scheduled $j^*$; we call this schedule the final schedule.  We shall bound $\E[\tilde C_{j^*}]/C_{j^*}$ (recall that $\tilde C_{j^*}$ is the completion time of $j^*$ in the schedule we output), by bounding the total length of idle and busy slots in the final schedule before $\tilde C_{j^*}$ separately. Recall that a time point is busy if all the $m$ machines are processing some jobs at that time, and idle otherwise. The next lemma gives this bound for a fixed $\theta$. The first (resp. second) term on the right side bounds the total length of busy (resp. idle) slots before $\tilde C_{j^*}$. The clean bound $a_{j^*}$ on the total length of idle slots comes from the unit-job size property.
	
	\begin{lemma}
		\label{lemma:PwC-uniform-bound-C-jstar}
		$\displaystyle \tilde C_{j^*} \leq \frac{1}{m}\Big|\set{j:  M^\theta_{j} \leq M^\theta_{j^*} } \Big| + a_{j^*}.$
	\end{lemma}
\ifdefined\CR
\else
	\begin{proof}
		The total length of busy slots in $(0, \tilde C_{j^*}]$ is at most $1/m$ times the total number of jobs scheduled so far, which is at most $\frac{1}{m}\Big|\set{j:  M^\theta_{j} \leq M^\theta_{j^*} } \Big|$.
		
		We now bound the total length of idle slots in $(0, \tilde C_{j^*}]$.  We start from $j_1 = j^*$. For every $\ell = 1, 2, 3 \cdots$, let $j_{\ell+1} \prec j_{\ell}$ be the job that is scheduled the last in the final schedule; if $j_{\ell+1}$ does not exists, then we let $a = \ell$ and break the loop. Thus, we constructed a chain $j_a \prec j_{a-1} \prec j_{a-2} \prec \cdots \prec j_1 = j^*$ of jobs.  By the definition of $a_{j^*}$, we have $a \leq a_{j^*}$.
		
		We shall show that for every idle unit slot $(t - 1, t]$ in $(0, \tilde C_{j^*}]$, some job in the chain is scheduled in $(t - 1, t]$. Assume otherwise; let $\ell \in [a]$ be the largest index such that $\tilde C_{j_\ell} > t$ ($\ell$ exists since $j_1 = j^*$ is scheduled in $(\tilde C_{j^*}-1, \tilde C_{j^*}]$ and $t < \tilde C_{j^*}$).  Then either $j_{\ell+1}$ does not exist, or $\tilde C_{j_{\ell+1}} \leq t-1$. Consider the iteration when $j_\ell$ is considered in the list scheduling algorithm. Before the iteration, $(t-1, t]$ is available for scheduling.  Since $j_\ell$ is scheduled after $t$, there must be a job $j' \prec j_{\ell}$ such that $\tilde C_{j'} \geq t$. Thus, we would have set $j_{\ell + 1} = j'$, a contradiction.
		
		Thus, the idle slots before $\tilde C_{j^*}$ are covered by the scheduling intervals of jobs in the chain. So, the total number of idle slots before $\tilde C_{j^*}$ is at most $a \leq a_{j^*}$. Overall, we have $\tilde C_{j^*} \leq \frac{1}{m}\Big|\set{j:  M^\theta_{j} \leq M^\theta_{j^*} } \Big| + a_{j^*}$.
	\end{proof}		
\fi

	We shall use $g(\theta) = M^\theta_{j^*}$ for every $\theta \in (0, 1]$. Notice that $C_{j^*} = \int_{\theta = 0}^{1}g(\theta) \sfd \theta$. For simplicity, let $g(0) = \lim_{\theta \to 0^+}g(\theta)$; so, $g(0)$ will be the smallest $t$ such that $x_{j^*, t} > 0$.   By Claim~\ref{claim:chain}, we have $g(0) \geq a_{j^*}$. For every $j \in J$ and $\theta \in [0, 1]$, define $\displaystyle h_j(\theta):= \sum_{t = 1}^{g(\theta)} x_{j, t}$. This is the total volume of job $j$ scheduled in $(0, g(\theta)]$. Thus, we have $\sum_{j \in J} h_j(\theta) \leq g(\theta)$. Noticing that $M^\theta_j \leq M^\theta_{j^*}$ if and only if $h_j(\theta) \geq \theta$. So, by Lemma~\ref{lemma:PwC-uniform-bound-C-jstar}, we have $\tilde C_{j^*} \leq g(0) + \frac1m\sum_{j \in J}\mathbf{1}_{h_j(\theta) \geq \theta}$. Thus, we can bound $\frac{\E\left[\tilde C_{j^*}\right]}{C_{j^*}}$ by the superior of 
	\begin{align}
		\frac{g(0) +  \frac1m\sum_{j \in J}\int_{\theta=0}^1\mathbf{1}_{h_j(\theta)\geq \theta} \sfd \theta}{\int_{\theta=0}^1g(\theta)\sfd\theta} \label{inequ:PwC-uniform-superior}
	\end{align} 
	subject to 
	\begin{itemize}
		\renewcommand{\theenumi}{\ref{inequ:PwC-uniform-superior}.\arabic{enumi}}
		\item $g:[0, 1] \to [1, \infty)$ is piecewise linear, left-continuous and non-decreasing, \refstepcounter{enumi} \hfill (\theenumi)\label{property:g-increasing}
		\item $\forall j \in J$, $h_j:[0,1]\to[0,1]$ is piecewise linear, left-continuous and non-decreasing, \refstepcounter{enumi} \hfill (\theenumi)\label{property:h-increasing}
		\item 	$\displaystyle	\sum_{j \in J} h_j(\theta) \leq  mg(\theta), \quad \forall \theta \in [0, 1]$.\refstepcounter{enumi} \hfill (\theenumi)\label{property:h-less-than-g}
	\end{itemize}
	
	To recover the 3-approximation ratio of \cite{MQS98}, we know that $\displaystyle g(0)\Big/\int_{\theta = 0}^1g(\theta) \sfd\theta \leq 1$; this corresponds to the fact $a_{j^*} \leq \tilde C_{j^*}$. It is not hard to show $\displaystyle \frac1m\sum_{j \in J}\int_{\theta=0}^1\mathbf{1}_{h_j(\theta)\geq \theta} \sfd \theta\Big/\left(\int_{\theta=0}^1g(\theta)\sfd\theta\right)\leq 2$. The tight factor $2$ can be achieved when $h_j(\theta) = \theta$ for every $j \in J$ and $\theta \in [0, 1]$ and $g(\theta) = n\theta/m$.  This corresponds to the following case: by the time $\theta$-fraction of job $j^*$ was completed, exactly $\theta$ faction of every job $j \in J$ was completed.  However, the two bounds can not be tight simultaneously: the first bound being tight requires $g$ to be a constant function, where the second bound being tight requires $g$ to be linear in $\theta$. This is where we obtain our improved approximation ratio. 

\ifdefined\CR
	Due to the page limit, we shall defer the proof that \eqref{inequ:PwC-uniform-superior} is at most $1+\sqrt{2}$ to the full version of the paper.  Here we only 
\else
	Before formally proving that \eqref{inequ:PwC-uniform-superior} is at most $1+\sqrt{2}$, we 
\fi
	give the combination of $g$ and $\set{h_j}_{j \in J}$ achieving the bound; the way we prove the upper bound is by showing that this combination is the worst possible. In the worst case, we have $h_j(\theta) = \max\set{\alpha, \theta}$ for every $\theta \in [0, 1]$, where $\alpha = \sqrt{2} - 1$. $g(\theta) = \frac{1}{m}\sum_{j \in J}h_j(\theta) = \frac{n}{m} \max\set{\alpha, \theta}$ for every $\theta \in [0, 1]$. In this case, the numerator of \eqref{inequ:PwC-uniform-superior} is $g(0) + \frac1m\sum_{j \in J}\int_{\theta = 0}^1\mathbf1_{h_j(\theta)\geq \theta}\sfd\theta  = (\alpha + 1)\frac{n}{m}$; the denominator is  $\int_{\theta  = 0}^1g(\theta)\sfd \theta = \frac{(1+\alpha^2)n}{2m}$. Thus, \eqref{inequ:PwC-uniform-superior} in this case is $\frac{\alpha + 1}{(1+\alpha^2)/2} = 1 + \sqrt{2}$. 

\ifdefined\CR
\else
	\paragraph{Computing the superior of \eqref{inequ:PwC-uniform-superior}} We now compute the superior of \eqref{inequ:PwC-uniform-superior} subject to Constraints~\eqref{property:g-increasing}, \eqref{property:h-increasing} and \eqref{property:h-less-than-g}. The functions $g$ and $h_j$'s defined may have other stronger properties (e.g, they are piecewise constant functions); however in the process, we only focus on the properties described above. In the following, $j$ in a summation is over all jobs in $J$. 
	
	First, we can add the following constraint without changing the superior of \eqref{inequ:PwC-uniform-superior}.
	\begin{align}
	\exists \theta^* \in [0, 1], \text{ such that } g(\theta) = g(0), \forall \theta \in [0, \theta^*], \text{ and } g(\theta) = \frac1m\sum_{j}h_j(\theta), \forall \theta \in (\theta^*, 1].
	\label{property:g-is-max-of-2}
	\end{align}
	
	To see this, we take any $g$ and $\{h_j\}_j$ satisfying Constraints~\eqref{property:g-increasing} to \eqref{property:h-less-than-g}. 
	%
	Let $g'(\theta) = \max\set{g(0), \frac1m\sum_{j}h_j(\theta)}$ for every $\theta \in [0, 1]$. Then $g':[0,1]\to [0, \infty)$ is piecewise-linear, left-continuous and non-decreasing because of Property~\eqref{property:h-increasing}; so $g'$ satisfies Property~\eqref{property:g-increasing}. Obviously $g'$ satisfies Property~\eqref{property:h-less-than-g}.  We define $\theta^* = \sup\set{\theta \in [0, 1]: g'(\theta) = g'(0)}$. Since $g'$ is left-continuous and monotone non-decreasing,  we have $g'(\theta) = g'(0) \geq \frac1m\sum_{j}h_j(\theta)$ for every $\theta \in [0, \theta^*]$ and $g'(\theta) = \frac1m\sum_{j} h_j(\theta) > g'(0)$ for every $\theta \in (\theta^*, 1]$. Thus, $g'$ satisfies Constraint~\eqref{property:g-is-max-of-2}. Changing $g$ to $g'$ will not decrease \eqref{inequ:PwC-uniform-superior}: the numerator does not change and the denominator can only decrease since $g(\theta) \geq \max\set{g(0),\frac{1}{m}\sum_{j}h_j{\theta}} = g'(\theta)$ for every $\theta \in [0, 1]$.  Thus, we can impose Constraint~\eqref{property:g-is-max-of-2}, without changing the superior of \eqref{inequ:PwC-uniform-superior}.
	
	
	Then, we can make the following constraint on $\{h_j\}_j$:
	\begin{align}
		\text{For every $j \in J$,  $h_j$ is a constant over $[0, \theta^*]$.}
	\end{align}
	
	For every $j$, we define $h'_j(\theta) = h_j(\theta^*)$ if $\theta \in [0, \theta^*]$ and $h'_j(\theta) = h_j(\theta) $ if $\theta \in (\theta^*, 1]$. Then each $h'_j$ satisfies Constraint~\eqref{property:h-increasing}.  Moreover, $\{h'_j\}_j$ satisfies Constraint ~\eqref{property:h-less-than-g} since $\sum_{j}h'_j(\theta) = \sum_{j} h_j(\theta^*) \leq mg(\theta^*) = mg(\theta)$ for every $\theta \in [0, \theta^*]$. Moreover, if we change $h_j$ to $h'_j$ for every $j \in J$, the denominator of \eqref{inequ:PwC-uniform-superior} does not change and the numerator can only increase.
	
	Finally, we can assume 
	\begin{align}
		\sum_{j}h_j(\theta) = mg(\theta), \quad \forall \theta \in [0, \theta^*].
	\end{align}
	Notice that the functions $\set{h_j}_{j}$ and $g$ are constant functions over $[0, \theta^*]$ and $\sum_{j}h_j(\theta) = mg(\theta)$ for $\theta \in (\theta^*, 1]$. If the constraint is not satisfied, we can find some $j$ such that $h_j$ is not right-continuous at $\theta^*$. Then, we increase $h_j(\theta)$ simultaneously for all $\theta \in [0, \theta^*]$ until the constraint is satisfied or $h_j$ becomes right-continuous at $\theta^*$. This process can be repeated until the constraint becomes satisfied. 
	
	Thus, with all the constraints, we have $g(\theta) = \frac1m \sum_j h_j(\theta)$ for every $\theta \in [0, 1]$. So, \eqref{inequ:PwC-uniform-superior} becomes 
	\begin{align*}
	\frac{\frac1m\sum_{j} h_j(0) + \frac1m\sum_j \int_{\theta = 0}^1 \mathbf{1}_{h_j(\theta)\geq \theta} \sfd \theta }{\frac1m\sum_{j}\int_{\theta  = 0 }^1 h_j(\theta)\sfd \theta} = \frac{\sum_{j} \left(h_j(0) + \int_{\theta = 0}^1 \mathbf{1}_{h_j(\theta)\geq \theta} \sfd \theta \right)}{\sum_{j}\int_{\theta  = 0 }^1 h_j(\theta)\sfd \theta}.
	\end{align*}
	
	To bound the quantity, it suffices to upper bound 
	\begin{align}
		\sup_h \frac{h(0) + \int_{\theta = 0}^1\mathbf{1}_{h(\theta) \geq \theta}\sfd \theta}{\int_{\theta=0}^{1}h(\theta)\sfd \theta }, \label{inequ:bound-in-terms-of-h}
	\end{align}
	where the superior is over all piecewise-linear, left-continuous and monotone non-decreasing functions $h:[0, 1] \to [0, 1]$.

	\begin{claim}
		Let $\alpha = h(0)$ and $\beta = \int_{\theta = 0}^{1}\mathbf1_{h(\theta) \geq \theta}\sfd \theta$. Then $0 \leq \alpha \leq \beta \leq 1$ and $\displaystyle \int_{\theta  = 0}^{1}h(\theta)\sfd \theta \geq \frac{\alpha^2}{2} + \beta - \frac{\beta^2}{2}$.
	\end{claim}
	
	\begin{proof}
		$\beta \geq \alpha$ because $h(\theta) \geq \theta$ holds for every $\theta \in [0, \alpha]$. To prove the second part, we can assume that $h(\theta) = \alpha$ if $\theta \in [0, \alpha]$: otherwise, we can change $h(\theta)$ to $\alpha$ for every $\theta \in (0, \alpha]$; this does not change $\int_{\theta = 0}^{1}\mathbf1_{h(\theta) \geq \theta}\sfd \theta$ and can only decrease $\int_{\theta = 0}^1h(\theta)\sfd \theta$. Also, we can assume that $h(\theta) \leq \theta$ for every $\theta \in [\alpha, 1]$. Otherwise, for every $\theta$ such that $h(\theta) > \theta$, we decrease $h(\theta)$ to $\theta$. This does not change $\int_{\theta = 0}^{1}\mathbf1_{h(\theta) \geq \theta}\sfd \theta$ and can only decrease $\int_{\theta = 0}^1h(\theta)\sfd \theta$. 
		\begin{align*}
			\int_{\theta = 0}^{1}h(\theta)\sfd\theta = \alpha^2 + \int_{\theta = \alpha}^1\big(\theta - (\theta - h(\theta))\big)\sfd\theta = \frac12 + \frac{\alpha^2}{2} - \int_{\theta = \alpha}^1(\theta - h(\theta))\sfd\theta.
		\end{align*}
		For interval $[a, b] \subseteq [\theta^*, 1]$ such that $h(a) = a$, we have that $\int_{\theta = a }^b(\theta - h(\theta))\sfd\theta \leq \frac{(b-a)^2}{2}$. Since $\int_{\theta = \alpha}^1\mathbf{1}_{h(\theta) < \theta} \leq 1-\beta$, we have that $\int_{\theta = \alpha}^1 (\theta - h(\theta)) \leq \frac{(1-\beta)^2}{2}$. Thus, the above quantity is at least $\frac12 + \frac{\alpha^2}{2} - \frac{(1-\beta)^2}{2}  = \frac{\alpha^2}{2}  + \beta - \frac{\beta^2}{2}$.
	\end{proof}
	
	Thus, \eqref{inequ:bound-in-terms-of-h} is at most $\sup_{0 \leq \alpha \leq \beta \leq 1}\frac{\alpha + \beta}{\frac{\alpha^2}{2} + \beta - \frac{\beta^2}{2}} = \sup_{0 \leq \alpha \leq \beta \leq 1}\frac{2(\alpha + \beta)}{2\beta - (\alpha+\beta)(\beta - \alpha)}$. Scaling both $\alpha$ and $\beta$ up can only increase $\frac{2(\alpha + \beta)}{2\beta - (\alpha+\beta)(\beta - \alpha)}$. Thus, we can assume $\beta = 1$ and \eqref{inequ:bound-in-terms-of-h} becomes $\sup_{\alpha \in [0, 1]}\frac{2(1+\alpha)}{1+\alpha^2}$. For $\alpha \in [0, 1]$, $\frac{2(1+\alpha)}{1+\alpha^2}$ is maximized at $\alpha^* = \sqrt{2}-1$ and the maximum value is $\frac{2(1+\sqrt{2}-1)}{1+(\sqrt{2}-1)^2} = \sqrt{2} + 1$. This finishes the proof of the $(\sqrt{2}+1)$-approximation for $P|\tprec,p_j=1|\sum_{j} w_j C_j$ (Theorem~\ref{theorem:main-PprecwCuniform}). 
\fi
	
	\section{Scheduling on Related Machines with Job Precedence Constraints}
\label{sec:QwC}

In this section, we give our $O(\log m /\log \log m)$-approximation for $Q|\tprec|C_{\max}$  and $Q|\tprec|\sum_{j}w_jC_j$, proving Theorem~\ref{theorem:main-QprecwC}. 
This slightly improves the previous best $O(\log m)$-approximation, due to Chudak and Shmoys \cite{CS97}. 
Our improvement comes from a better tradeoff between two contributing factors.

As in \cite{CS97}, we can convert the objective of minimizing  total weighted completion time to minimizing  makespan, losing a factor of 16. 
We now describe the LP used in \cite{CS97} and state the theorem for the reduction. Throughout this section, $i$ is restricted to machines in $M$, and $j$ and $j'$ are restricted to jobs in $J$. 

\ifdefined\CR
	\begin{equation}
	\min \qquad D \tag{$\text{LP}_{\text{Q}|\text{prec}|\text{Cmax}}$} \label{LP:QCmax}
	\end{equation}	\vspace*{-25pt}
	
		\begin{alignat}{2}
		\sum_{i}x_{i, j} &= 1  &\qquad &\forall j \label{LPC:QCmax-job-scheduled} \\
		p_j \sum_{i}\frac{x_{i, j}}{s_i} &\leq C_j &\qquad &\forall j \label{LPC:QCmax-completion-time} \\
		C_j + p_{j'}\sum_{i}\frac{x_{i, j'}}{s_{i}} &\leq C_{j'} &\qquad &\forall j, j', j \prec j' \label{LPC:QCmax-precedence} \\
		\frac{1}{s_i}\sum_{j} p_j x_{i, j} &\leq D &\qquad &\forall i\label{LPC:QCmax-machine-capacity} \\
		C_j &\leq D &\qquad &\forall j \label{LPC:QCmax-C-lt-D} \\
		x_{i, j}, C_j &\geq 0  &\qquad &\forall j, i \label{LPC:QCmax-non-negative} 
		\end{alignat}
\else
	\begin{equation}
		\min \qquad D \tag{$\text{LP}_{\text{Q}|\text{prec}|\text{Cmax}}$} \label{LP:QCmax}
	\end{equation}	\vspace*{-20pt}
	
	\noindent\begin{minipage}{0.53\textwidth}
		\begin{alignat}{2}
			\sum_{i}x_{i, j} &= 1  &\qquad &\forall j \label{LPC:QCmax-job-scheduled} \\
			p_j \sum_{i}\frac{x_{i, j}}{s_i} &\leq C_j &\qquad &\forall j  \label{LPC:QCmax-completion-time} \\
			C_j + p_{j'}\sum_{i}\frac{x_{i, j'}}{s_{i}} &\leq C_{j'} &\qquad &\forall j, j', j \prec j' \label{LPC:QCmax-precedence}
		\end{alignat}
	\end{minipage}
	\begin{minipage}{0.47\textwidth} \vspace*{3pt}
		\begin{alignat}{2}
			\frac{1}{s_i}\sum_{j} p_j x_{i, j} &\leq D &\qquad &\forall i \label{LPC:QCmax-machine-capacity} \\
			C_j &\leq D &\qquad &\forall j  \label{LPC:QCmax-C-lt-D} \\[13pt]
			x_{i, j}, C_j &\geq 0  &\qquad &\forall j, i \label{LPC:QCmax-non-negative} \\[3pt] \nonumber
		\end{alignat}
	\end{minipage} \medskip
\fi

	\eqref{LP:QCmax} is a valid LP relaxation for $Q|\tprec|C_{\max}$. In the LP, $x_{i, j}$ indicates whether job $j$ is scheduled on machine $i$.  $D$ is the makespan of the schedule, and $C_j$ is the completion time of $j$ in the schedule. Constraint~\eqref{LPC:QCmax-job-scheduled} requires every job $j$ to be scheduled. Constraint~\eqref{LPC:QCmax-completion-time} says that the completion time of $j$ is at least the processing time of $j$ on the machine it is assigned to. Constraint~\eqref{LPC:QCmax-precedence} says that if $j \prec j'$, then $C_{j'}$ is at least $C_j$ plus the processing time of $j'$ on the machine it is assigned to. Constraint~\eqref{LPC:QCmax-machine-capacity} says that the makespan $D$ is at least the total processing time of all jobs assigned to $i$, for every machine $i$. Constraint~\eqref{LPC:QCmax-C-lt-D} says that the makespan $D$ is at least the completion time of any job $j$.  Constraint~\eqref{LPC:QCmax-non-negative} requires the $x$ and $C$ variables to be non-negative. 
	
	The value of \eqref{LP:QCmax} provides a lower bound on the makespan of any valid schedule. However, even if we require each $x_{i, j} \in \set{0, 1}$, the optimum solution to the integer programming is not necessarily a valid solution to the scheduling problem, since it does not give a scheduling interval for each job $j$. Nevertheless, we can use the LP relaxation to obtain our $O(\log m/\log\log m)$-approximation for $Q|\tprec|C_{\max}$.  Using the following theorem from \cite{CS97}, we can extend the result to $Q|\tprec|\sum_{j}w_jC_j$: 
	\begin{theorem}[\cite{CS97}]
		\label{theorem:QwC-reduction}
		Suppose there is an efficient algorithm $\calA$ that can round a fractional solution to \eqref{LP:QCmax} to a valid solution to the correspondent $Q|\tprec|C_{\max}$ instance, losing only a factor of $\alpha$. Then there is a $16\alpha$-approximation for the problem $Q|\tprec|\sum_j w_j C_j$.
	\end{theorem}
	
	Thus, from now on, we focus on the objective of minimizing the makespan; our goal is to design an efficient rounding algorithm as stated in Theorem~\ref{theorem:QwC-reduction} with $\alpha = O(\log m/\log\log m)$.  We assume that $m$ is big enough. For the given instance of $Q|\tprec|C_{\max}$, we shall first pre-processing the instance as in \cite{CS97} so that it contains only a small number of groups.  In the first stage of the pre-processing step, we discard all the machines whose speed is at most $1/m$ times the speed of the fastest machine.  Since there are $m$ machines, the total speed for discarded machines is at most the speed of the fastest machine.  In essence, the fastest machine can do the work of all the discarded machines; this will increase the makespan by a factor of 2. Formally, let $i^*$ be the machine with the fastest speed. For every discarded machine $i$ and any job $j$ such that $x_{i, j} > 0$, we shall increase $x_{i^*, j}$ by $x_{i, j}$ and change this $x_{i, j}$ to $0$. By scaling $D$ by a factor of $2$, the LP solution remains feasible. To see this, notice that the modification to the fractional solution can only decrease $p_j\sum_i \frac{x_{i,j}}{s_i}$ for each $j$. The only constraint we need to check is Constraint~\eqref{LPC:QCmax-machine-capacity} for $i = i^*$. Since $\sum_{i\text{ discarded}, j}\frac{p_jx_{i,j}}{s_{i^*}}  \leq \sum_{i \text{ discarded},j} \frac{p_jx_{i,j}}{ms_i} \leq \sum_{i \text{ discarded}} \frac{D}{m} \leq D$, moving the scheduling of jobs from discarded machines to $i^*$ shall only increase the processing time of jobs on $i^*$ by $D$.  Thus, we assume all machines have speed larger than $1/m$ times the speed of the fastest machine.  By scaling speeds of machines uniformly, we assume all machines $i$ have speed $i \in [1, m)$, and $|M| \leq m$.
	
	In the second stage of the pre-processing step, we partition the machines into groups, where each group contains machines with similar speeds. Let $\gamma = \log m/\log\log m$.  Then group $M_k$ contains machines with speed in $[\gamma^{k-1}, \gamma^{k})$, where $k = 1, 2, \cdots, K := \ceil{\log _\gamma m} = O(\log m/\log\log m)$.   We remark that that unlike \cite{CS97}, we \emph{can not} round down the speed of each machine $i$ to the nearest power of $\gamma$. If we do so, we will lose a factor of $(\log m/\log\log m)$ and finally we can only obtain an $O((\log m/\log\log m)^2)$-approximation.  Instead, we keep the speeds of machines unchanged.

	We now define some useful notations. For a subset $M' \subseteq M$ of machines, we define $s(M') = \sum_{i \in M'}s_i$ to be the total speed of machines in $M'$; for $M' \subseteq M$ and $j \in J$, let $x_{M', j} = \sum_{i \in M'}x_{i, j}$ be the total fraction of job $j$ assigned to machines in $M'$. 
	
	For any job $j$, let $\ell_j$ be the largest integer $\ell$ such that $\sum_{k = \ell}^{K}x_{M_k, j} \geq 1/2$. That is, the largest $\ell$ such that at least $1/2$ fraction of $j$ is assigned to machines in groups $\ell$ to $K$.  Then, let $k_j$ be the index $k \in [\ell_j, K]$ that maximizes $s(M_k)$.  That is, $k_j$ is the index of the group in groups $\ell_j$ to $K$ with the largest total speed.  Later in the machine-driven list scheduling algorithm, we shall constrain that job $j$ can only be assigned to machines in group $k_j$.  The following claim says that the time of processing $j$ on any machine in $M_{k_j}$ is not too large, compared to processing time of $j$ in the LP solution. 
	
	\begin{claim}
		\label{claim:QwC-factor-for-idle}
		For every $j \in J$, and any machine $i \in M_{k_j}$, we have $\displaystyle \frac{p_j}{s_i} \leq 2\gamma \sum_{i' \in M}\frac{p_j  x_{i', j}}{s_{i'}}$.
	\end{claim}
	\begin{proof}
		Notice that $\displaystyle \sum_{k = \ell_j + 1}^K x_{M_k, j} < 1/2$ by our definition of $\ell_j$.  Thus, $\displaystyle \sum_{k = 1}^{\ell_j} x_{M_k, j} > 1/2$.  Then, $\displaystyle \sum_{i' \in M}\frac{x_{i', j}}{s_{i'}} \geq \sum_{i' \in \union_{k=1}^{\ell_j}M_k} \frac{x_{i', j}}{s_{i'}} \geq \frac{1}{2} \cdot \gamma^{-\ell_j}$. This is true since $\displaystyle \sum_{i' \in \union_{k=1}^{\ell_j}M_k} x_{i', j} \geq 1/2$ and every $i'$ in the sum has $\frac{1}{s_{i'}} \geq \gamma^{-\ell_j}$.
		
		Since $i$ is in group $k_j \geq \ell_j$,  $i'$ has speed at least $\gamma^{\ell_j - 1}$ and thus $\frac{1}{s_i} \leq \gamma ^{1-\ell_j}$.  Then the claim follows.
	\end{proof}
	
	\begin{claim}
		\label{claim:QwC-factor-for-busy}
		$\displaystyle \sum_{j \in J}\frac{p_j}{s(M_{k_j})} \leq 2KD$.
	\end{claim}
	
	\begin{proof}
		Focus on each job $j \in J$. Noticing that $\displaystyle \sum_{k=\ell_j}^K x_{M_k, j} \geq 1/2$, and $k_j$ is the index of the group with the maximum total speed, we have
		\begin{align*}
			\sum_{k = 1}^K \frac{x_{M_k, j}}{s(M_k)} \geq \sum_{k = \ell_j}^K\frac{x_{M_k, j}}{s(M_k)} \geq \frac{1}{2s(M_{k_j})}.
		\end{align*}
		Summing up the above inequality scaled by $2p_j$, over jobs $j$, we have
\ifdefined\CR
		\begin{align*}
			&\quad \sum_{j \in J} \frac{p_j}{s(M_{k_j})} \leq 2\sum_{j \in J}p_j \sum_{k = 1}^K\frac{x_{M_k, j}}{s(M_k)} \\
			&= 2\sum_{k = 1}^K\frac{1}{s(M_k)}\sum_{j \in J}p_jx_{M_k, j} \leq 2\sum_{k = 1}^K D = 2KD.
		\end{align*}
\else
		\begin{align*}
			\sum_{j \in J} \frac{p_j}{s(M_{k_j})} \leq 2\sum_{j \in J}p_j \sum_{k = 1}^K\frac{x_{M_k, j}}{s(M_k)} = 2\sum_{k = 1}^K\frac{1}{s(M_k)}\sum_{j \in J}p_jx_{M_k, j} \leq 2\sum_{k = 1}^K D = 2KD.
		\end{align*}
\fi
		To see the last inequality, we notice that $\sum_{j \in J}p_jx_{M_k, j}$ is the total size of jobs assigned to group $k$, $s(M_k)$ is the total speed of all machines in $M_k$ and $D$ is the makespan. Thus, we have $\sum_{j \in J}p_jx_{M_k, j} \leq s(M_k)D$. Formally, Constraint~\eqref{LPC:QCmax-machine-capacity} says $\sum_{j \in J} p_j x_{i, j} \leq s_iD$ for every $i \in M_k$. Summing up the inequalities over all $i \in M_k$ gives $\sum_{j \in J}p_jx_{M_k, j} \leq s(M_k)D$. 
	\end{proof}
	
	With the $k_j$ values, we can run the machine-driven list-scheduling algorithm in \cite{CS97}. The algorithm constructs the schedule in real time. Whenever a job completes (or at the beginning of the algorithm), for each idle machine $i$,  we attempt to schedule an unprocessed job $j$ on $i$ subject to two constraints: (i) machine $i$ can only pick a job $j$ if $i \in M_{k_j}$ and (ii) all the predecessors of $j$ are completed. If no such job $j$ exists, machine $i$ remains idle until a new job is competed.   We use $\calS$ to denote this final schedule; Let $i_j \in M_{k_j}$ be the machine that process $j$ in the schedule constructed by our algorithm. 
	
	The following simple observation is the key to prove our $O(\log m/\log\log m)$-approximation. Similar observations were made and used implicitly in \cite{CS97}, and in \cite{Gra69} for the problem on identical machines. However, we think stating the observation in our way makes the analysis cleaner and more intuitive.   We say a time point $t$ is critical, if some job starts or ends at $t$.  To avoid ambiguity, we exclude these critical time points from our analysis (we only have finite number of them).  At any non-critical time point $t$ in the schedule, we say a job $j$ is minimal if all its predecessors are completed but $j$ itself is not completed yet. 
	
	\begin{observation} 
		\label{obs:QwC-reduce-or-busy}
		At any non-critical time point $t$ in $\calS$, either all the minimal jobs $j$ are being processed, or there is a group $k$ such that all machines in $M_k$ are busy.
	\end{observation}
	\begin{proof}
		All the minimum jobs at $t$ are ready for processing. If some such job $j$ is not processed at $t$, it must be the case that all machines in $M_{k_j}$ are busy. 		
	\end{proof}
	
	As time goes in $\calS$, we maintain the precedence graph over $J'$, the set of jobs that are not completed yet: we have an edge from $j \in J'$ to $j' \in J'$ if $j \prec j'$.  At any time point, the weight of a job $j$ is the time needed to complete the rest of job $j$ on $i_j$, i.e, the size of the unprocessed part of job $j$, divided by $s_{i_j}$. If at $t$, all minimum jobs are being processed, then the weights of all minimal jobs are being decreased at a rate of 1. Thus, the length of the longest path of in the precedence graph is being decreased at a rate of $1$. The total length of the union of these time points is at most length of the longest path in the precedence graph at time $0$, which is at most 
	\begin{align*}
		\max_{H} \sum_{j \in H} \frac{p_j}{s_{i_j}} \leq \max_{H} 2\gamma \sum_{j \in H}\sum_{i \in M}\frac{p_jx_{i, j}}{s_i} \leq 2\gamma D,
	\end{align*}
	where $H$ is over all precedence chains of jobs.  The first inequality is by Claim~\ref{claim:QwC-factor-for-idle} and the second inequality is by Constraints~\eqref{LPC:QCmax-precedence} and \eqref{LPC:QCmax-C-lt-D} in the LP. 
	
	If not all the minimal jobs are being processed at time $t$, then there must be a group $k$ such that all machines in group $k$ are busy, by Observation~\ref{obs:QwC-reduce-or-busy}. The total length of the union of all these points is at most 
	\begin{align*}
		\sum_{k}\frac{\sum_{j: k_j = k}p_j}{s(M_k)} = \sum_{j \in J}\frac{p_j}{s(M_{k_j})} \leq 2KD,
	\end{align*}
	by Claim~\ref{claim:QwC-factor-for-busy}. 
	
	Thus, our schedule has makespan at most $2(\gamma + K)D = O(\log m/\log\log m)D$, leading to an $O(\log m/\log\log m)$-approximation for $Q|\tprec|C_{\max}$. Combining this with Theorem~\ref{theorem:QwC-reduction}, 
	we obtain an $O(\log m/\log \log m)$-approximation for $Q|\tprec|\sum_j w_jC_j$, finishing the proof of Theorem~\ref{theorem:main-QprecwC}. Indeed, as shown in \cite{CS97}, this factor is tight if we use \eqref{LP:QCmax}.

	\newcommand{\cone}{0.2}  
\newcommand{\ctwo}{0.4}   
\newcommand{\cthree}{1.2}  
\newcommand{\cfour}{1.4}    
\newcommand{\cfive}{0.1}     
\newcommand{\csix}{0.01}

\newcommand{\bone}{10}

\newcommand{\assignmachine}{{\to}}
\newcommand{\assignblock}[1]{{\overset{#1}{\leadsto}}}
\newcommand{\samegroup}[1]{{\overset{#1}{\sim}}}
\newcommand{\notsamegroup}[1]{{\overset{#1}{\not\sim}}}

\section{$(1.5-c)$-Approximation Algorithm for Unrelated Machine Scheduling Based on Time-Indexed LP}
\label{sec:RwC}

In this section, we prove Theorem~\ref{theorem:main-RwC}, 
using the dependence rounding scheme of \cite{BSS16} as a black-box.    We first describe the rounding algorithm and then give the analysis. 

\subsection{Rounding Algorithm}
Let $x$ be a feasible solution to ~\eqref{LP:RwC} as stated in Theorem~\ref{theorem:main-RwC}. Again, we use $C_j = \sum_{i, s}x_{i, j, s} (s+p_{i,j})$ to denote the completion time of $j$ in the LP solution; thus the value of the LP is $\sum_{j} w_jC_j$.  Let $y_{i, j} = \sum_{s}x_{i, j, s}$ for every pair $i, j$ to denote the fraction of job $j$ that is scheduled on machine $i$; so $\sum_{i} y_{i, j} = 1$ for every $j$.   It is convenient to define a rectangle $R_{i, j, s}$ for every $x_{i, j, s} > 0$: $R_{i, j, s}$ has height $x_{i, j, s}$, with horizontal span being $(s, s+p_{i,j}]$.  We say $R_{i, j, s}$ is the rectangle for $j$ on machine $i$ at time $s$. We say a rectangle covers a time point (resp. a time interval), if its horizontal span covers the time point (resp. the time interval).

We can recover the classic 1.5-approximation for $R||\sum_{j}w_jC_j$ using \eqref{LP:RwC}. For each job $j \in J$, randomly choose a rectangle for $j$: the probability of choosing $R_{i, j, s}$ is $x_{i, j, s}$, i.e, its height. We shall assign job $j$ to $i$ and let $\tau_j$ be a random number in $(s, s+p_{i,j}]$. Then, all jobs $j$ assigned to machine $i$ will be scheduled in increasing order of their $\tau_j$ values.  To see this is a 1.5-approximation, fix a job $j^* \in J$ and condition on the event that $j^* \assignmachine i$ (indicating that $j^*$ is assigned to machine $i$) and $\tau_{j^*}$. Let $\tilde C_{j^*}$ be the completion time of $j^*$ in the schedule returned by the algorithm. Notice that $\E\left[\tilde C_{j^*}|j^* \assignmachine i, \tau_{j^*}\right] \leq \E\left[\sum_{j \neq j^*: j \assignmachine i, \tau_j \leq \tau_{j^*}}p_{i,j}|\tau_{j^*}\right] + p_{i, j^*}$. If for some $j \neq j^*$, the selected rectangle for $j$ is $R_{i, j, s}$ and $ \tau_j \leq \tau_{j^*}$, then we say the  $p_{i, j}$ term in summation inside $\E[\cdot]$ on the right side is contributed by the rectangle $R_{i,j,s}$. We then consider the contribution of each rectangle $R_{i, j, s}$, $j \neq j^*$ to the expectation on the right side. The probability that we choose the rectangle $R_{i, j, s}$ for $j$ is $x_{i, j, s}$. Under this condition, the probability that $\tau_j \leq \tau_{j^*}$ is exactly fraction of the portion of $R_{i, j, s}$ that is before $\tau_{j^*}$. When this happens, the contribution made by this $R_{i, j, s}$ is exactly $p_{i, j}$. Thus, the contribution of $R_{i, j, s}$ to the expectation is exactly the area of the portion of $R_{i, j, s}$  before $\tau_{j^*}$. Since the total height of rectangles on $i$ covering any time point is at most 1, the total contribution from all rectangles on $i$ is at most $\tau_{j^*}$. Thus $\E\left[\tilde C_{j^*}|j^* \assignmachine i, \tau_{j^*}\right] \leq \E\left[\sum_{j \neq j^*: j \assignmachine i, \tau_j \leq \tau_{j^*}}p_{i,j}|\tau_{j^*}\right] + p_{i, j^*} \leq \tau_{j^*} + p_{i, j^*}$.  Notice that conditioned on choosing rectangle $R_{i, j^*, s}$ for $j^*$, the expected value of $\tau_{j^*}$ is $s + p_{i, j^*}/2$. Thus, $\E\left[\tilde C_{j^*}\right] \leq \sum_{i, s}x_{i, j^*, s}(s + p_{i, j^*}/2 + p_{i, j^*}) = \sum_{i, s}x_{i, j^*, s}(s + 1.5p_{i, j^*}) \leq 1.5 \sum_{i, s}x_{i, j^*, s} (s+p_{i,j^*}) =1.5C_{j^*}$. This recovers the $1.5$-approximation ratio.

As \cite{BSS16} already showed, we can not beat 1.5 if the assignments of jobs to machines are independent. This lower bound is irrespective of the LP we use: even if the fractional solution is a convex combination of integral solutions for the problem, independently assigning jobs to machines with probabilities $\set{y_{i, j}}_{i,j}$ can only lead to a 1.5-approximation. To overcome this barrier, \cite{BSS16} used an elegant dependence rounding scheme, that is formally stated in the following theorem, which we shall apply as a black-box:

\begin{theorem}[\cite{BSS16}]
	\label{thm:BSS}
	Let $\zeta = 1/108$. Consider a bipartite graph $G = (M \cup J, E)$ between the set $M$ of machines and the set $J$ of jobs. Let $y \in [0, 1]^E$ be fractional values on the edges satisfying $y(\delta(j)) = 1$ for every job $j \in J$. For each machine $i \in M$, select any family of disjoint $E^1_i, E^2_i, \cdots, E^{\kappa_i}_i \subseteq \delta(i)$ subsets of edges incident to $i$ such that $y(E^\ell_{i}) \leq 1$ for $\ell = 1,\cdots, \kappa_i$. 
	
	Then, there exists a randomized polynomial-time algorithm that outputs a random subset of the edges $E^* \subseteq E$ satisfying
	\begin{enumerate}[label=(\alph*)]
		\item  For every $j \in J$ , we have $|E^* \cap \delta(j)| = 1$ with probability 1; 
		\item  For every $e \in E$, $\Pr[e\in E^*]=y_e$;
		\item For every $i \in M$ and all $e \neq e' \in \delta(i)$: 
		\begin{align*}
			\Pr\left [e\in E^*, e' \in E^*\right] \leq 
			\begin{cases}
				(1-\zeta)·\cdot y_e y_{e'} & \text{ if } \exists\ell \in [\kappa_i], e,e' \in E^{\ell}_i,\\
				y_ey_{e'} & otherwise.
			\end{cases}
		\end{align*}
	\end{enumerate}
\end{theorem}

In the theorem, $\delta(u)$ is the set of edges incident to the vertex $u$ in $G$, and $y(E') = \sum_{e \in E'}y_e$ for every $E' \subseteq E$.  We shall apply the theorem with $G = (M \cup J, E)$, with $E = \set{(i, j): y_{i, j} > 0}$ and $y$ values being our $y$ values.  In the theorem, we can specify a grouping for edges incident to every machine $i \in M$ subject to the constraint that the total $y$-value of all edges in a group is at most 1. The theorem says that we can select a subset $E^* \subseteq E$ of edges respecting the marginal probabilities $\set{y_{i, j}}_{(i, j) \in E}$, and satisfying the property that exactly one edge incident to any job $j$ is selected, and the negative correlation.  The key to the improved approximation ratio in \cite{BSS16} is that for two distinct edges $e, e'$ in the same group for $i$, their correlation is ``sufficiently negative''. 

The key to apply Theorem~\ref{thm:BSS} is to define a grouping for each machine $i$.  Notice that our analysis for the independent rounding algorithm suggests that the expected completion time for $j^*$ is at most $\sum_{i, s}x_{i, j^*, s}(s + 1.5p_{i, j^*})$; that is, there is no 1.5-factor before $s$. Thus, for the 1.5-approximation ratio to be tight, for most rectangles $R_{i, j^*, s}$, $s$ should be very small compared to $p_{i, j^*}$. Intuitively,  we define our random groupings such that, if such rectangles for $j$ and for $j'$ on machine $i$ overlap a lot, then $j$ and $j'$ will have a decent probability to be grouped together in the grouping for $i$. 

\paragraph{Defining the Groupings for Machines}  Our groupings for the machines are random. We choose a $\tau_{i, j}$ value for every job $j$ on every machine $i$ such that $y_{i, j} > 0$ (as opposed to choosing only one $\tau_j$ value for a job $j$ as in the recovered 1.5-approximation): choose $s_{i, j}$ at random such that $\Pr[s_{i, j}= s]=x_{i, j, s}/y_{i, j}$; this is well-defined since $\sum_{s}x_{i, j, s} = y_{i, j}$. Then we let $\tau_{i, j}$ be a random real number in $(s_{i, j}, s_{i,j}+p_{i,j}]$.   

Recall that the expected completion time of $j$ in the schedule given by the independence rounding is $\sum_{i, s}x_{i, j, s}(s + 1.5p_{i,j})$. As there is no $1.5$-factor before $s$, we can afford to ``shift'' a rectangle $R_{i,j,s}$ to the right side by a small constant times $s$ and use the shifted rectangles to sample $\set{s_{i,j}}_{i,j}$ and $\set{\tau_{i,j}}$ values.  On one hand, after the shifting of rectangles for $j$, $\E[\tilde C_j]$ can still be bounded by $1.5C_j$. On the other hand, the shifting of rectangles for $j$ will benefit the other jobs.   Formally, for every machine $i$ and job $j$ with $y_{i, j} > 0$, we define 
\begin{align*}
	\phi_{i, j} = \frac{1}{y_{i,j}}\sum_{s}x_{i,j,s}s 
\end{align*}
the average starting time of rectangles for job $j$ on machine $i$, and
\begin{align*}
	\theta_{i, j} = \cone (s_{i,j} + \phi_{i, j}) + \ctwo y_{i, j}p_{i, j}.
\end{align*}
to be the shifting parameter for job $j$ on $i$. Namely, we shall use the values of $\set{\tau_{i, j} + \theta_{i, j}}_{i,j}$ to decide the order of scheduling jobs. 

We shall distinguish between good jobs and bad jobs. Informally, we say a job $j$ is good on a machine $i$, if using the independent rounding algorithm, we can already prove a better than 1.5-factor on the completion time of $j$, conditioned on that $j$ is assigned to $i$. Formally,
\begin{definition}
	Given a job $j$ and machine $i$ with $y_{i, j} > 0$, we say $j$ is good on $i$ if 
	\begin{align*}
		\phi_{i, j} + y_{i, j}p_{i, j} \geq \csix  p_{i, j}.
	\end{align*}
	Otherwise, we say job $j$ is bad on $i$. 
\end{definition}

Returning to the recovered 1.5-approximation algorithm, we can show that $\E[\tilde C_{j^*} | j^* \assignmachine i] \leq (1.5-\Omega(1))(\phi_{i,j^*} + p_{i,j^*})$, if $j^*$ is good on $i$.  Either $\phi_{i, j^*} \geq 0.005p_{i,j^*}$ or $y_{i, j^*} \geq 0.005$. In the former case, we have $\E[\tilde C_{j^*} | j^* \assignmachine i] \leq \phi_{i, j^*} + 1.5p_{i,j^*}$, which is at most $(1.5 - \Omega(1))(\phi_{i,j^*}+p_{i,j^*})$. In the latter case, we total area of the portions of rectangles $\set{R_{i, j, s}:j \neq j^*, s}$ before $\tau_{j^*}$ is smaller than $\tau_{j^*}$ by $y_{i,j^*}p_{i, j^*}/2$, in expectation over all $\tau_{j^*}$. This also can save a constant factor.   The formal argument will be made in the proof of  Lemma~\ref{lemma:RwC-bound-good}, where we shall not use the strong negative correlation of the dependence rounding scheme. 

Now we are ready to define the groupings $\set{E^\ell_i}_{i\in M, j \in [\kappa_i]}$ in Theorem~\ref{thm:BSS}. Till this end, we fix a machine $i$ and show how to construct the grouping for $i$. If a job $j$ is good on $i$, then $(i, j)$ is not in any group. It suffices to focus on bad jobs on $i$; keep in mind that for these jobs $j$, both $\phi_{i, j}/p_{i, j}$ and $y_{i, j}$ are tiny; thus $s_{i,j}/p_{i,j}$ and $\theta_{i,j}/p_{i,j}$ will also be tiny with high probability.

\begin{definition}
	A basic block is a time interval $(2^a, 2^{a+1}] \subseteq (0, T]$, where $a \geq -2$ is an integer.  
\end{definition}

\begin{definition}
	\label{def:RwC-assign-job-to-block}
	For a bad job $j$ on machine $i$ (thus $y_{i, j}>0$), we say the edge $(i, j)$ is assigned to a basic block $(2^a, 2^{a+1}]$, denoted as $j \assignblock{i}a$, if 
	\begin{enumerate}[itemsep=0pt, topsep=3pt,label=(\ref{def:RwC-assign-job-to-block}\alph*), labelwidth=*, leftmargin=*]
		\item $(2^a, 2^{a+1}] \subseteq (\bone\phi_{i, j}, p_{i, j}]$, and \label{property:RwC-job-2-block-sub-interval}
		\item $s_{i, j} + \theta_{i, j}  \leq 2^a$, and \label{property:RwC-job-2-block-stheta-small}
		\item $\tau_{i, j} \in (2^a, 2^{a+1}]$.  \label{property:RwC-job-2-block-tau-inside}
	\end{enumerate}
\end{definition}

Property~\ref{property:RwC-job-2-block-sub-interval} requires the block to be inside $(10 \phi_{i, j}, p_{i, j}]$ and Property~\ref{property:RwC-job-2-block-tau-inside} requires $\tau_{i, j}$ to be inside the block.  Property~\ref {property:RwC-job-2-block-stheta-small} requires that for the rectangle for $j$ starting at $s_{i, j}$, after we shift it by $\theta_{i, j}$ distance, it still contains $(2^a, 2^{a+1}]$.  With Property~\ref{property:RwC-job-2-block-sub-interval} and the definition of $\phi_{i,j}$, we can prove the following lemma:
\begin{lemma}
	\label{lemma:RwC-weight-to-a-block}
	For every machine $i$ and a basic block $(2^a, 2^{a+1}]$, we always have $\sum_{j \assignblock{i} a}y_{i, j} \leq 10/9$.
\end{lemma}

\begin{proof}
	Focus on a job $j$ that is bad on $i$. We have $\sum_{s  \leq \bone\phi_{i, j}} x_{i, j,s} \geq 9y_{i, j}/10$, since otherwise we shall have $\sum_{s}x_{i, j, s} s > \bone\phi_{i, j} \cdot (y_{i, j}/10) = \phi_{i, j}y_{i,j}$, contradicting the definition of $\phi_{i, j}$.  Thus, $y_{i, j} \leq (10/9)\sum_{s \leq \bone\phi_{i, j}}x_{i, j, s}$.   If $j\assignblock{i}a$, then $(2^a, 2^{a+1}] \subseteq (\bone\phi_{i, j}, p_{i, j}]$ by Property~\ref{property:RwC-job-2-block-sub-interval}. Thus, $(2^a, 2^{a+1}]$ is covered by $(s, s+p_{i, j}]$ for every $s \leq \bone\phi_{i, j}$. Thus, we have $\sum_{j \assignblock{i} a} y_{i, j} \leq (10/9)\sum_{j \assignblock{i}a, s \leq \bone\phi_{i, j}}x_{i, j, s} \leq 10/9$.  The second inequality used the fact that the total height of rectangles covering $(2^a , 2^{a+1}]$ on machine $i$ is at most $1$. 
\end{proof}

For every basic block $(2^a, 2^{a+1}]$, we partition the set of edges assigned to $(2^a, 2^{a+1}]$ into at most 10 sets, each containing edges with total weight at most $1/8$. This is possible since the total weight of all edges assigned to $(2^a, 2^{a+1}]$ is at most $10/9$ by Lemma~\ref{lemma:RwC-weight-to-a-block}, and every bad job $j$ on $i$ has $y_{i, j} < \csix$: we can keep adding edges to a set until the total weight is at least $1/9$ and then we start constructing the next set; the number of sets we constructed is at most $(10/9)/(1/9) = 10$ and each set has a total weight of at most $1/9 + \csix  \leq 1/8$.  Now, we can randomly drop at most 2 sets so that we have at most $8$ sets remaining. Then we create a group for $i$ containing the edges in the remaining sets.   Thus, the total $y$-value of the edges in this group is at most $1$.

So, we have defined the grouping for $i$; recall that for each basic block $(2^a, 2^{a+1}] \subseteq (\bone\phi_{i, j}, p_{i,j}]$, we may create a group.  For a bad job $j$ on $i$, $(i, j)$ may not be assigned to any group. This may happen if $(i, j)$ is not assigned to any basic block, or if $(i, j)$ is assigned to a basic block, but we dropped the set containing $(i, j)$ when constructing the group for the basic block. 

\paragraph{Obtaining the Final Schedule} With the groupings for the machines, we can now apply Theorem~\ref{thm:BSS} to obtain a set $E^*$ of edges that satisfies the properties of the theorem.  If $(i, j) \in E^*$, then we assign $j$ to $i$; $j$ is assigned to exactly one machine since $E^*$ contains exactly one edge incident to $j$. For all the jobs $j$ assigned to $i$, we schedule them according to the increasing order of $\tau_{i, j} + \theta_{i, j}$; with probability 1, no two jobs $j$ have the same $\tau_{i, j} + \theta_{i, j}$ value.   Let $\tilde C_j$ be the completion time of $j$ in this final schedule.

\subsection{Analysis of Algorithm}


	\paragraph{Notations and Simple Observations}  From now on, we use $j \samegroup{i} j'$ to denote the event that $(i,j)$ and $(i, j')$ are assigned to the same group for $i$, in the grouping scheme of Theorem~\ref{thm:BSS}.  We use $j\assignmachine i $ to denote the event that $j$ is assigned to the machine $i$. Recall that $j \assignblock{i} a$ indicates the event that $(i, j)$ is assigned to the basic block $(2^a, 2^{a+1}]$. From the way we construct the groups, the following observation is immediate: 
	\begin{observation}
		\label{obs:RwC-prob-same-group}
		Let $(2^a, 2^{a+1}]$ be a basic block, $j \neq j'$ be bad jobs on $i$. We have 
		\begin{align*}
			\Pr\left[j \samegroup{i} j'\big|j\assignblock{i} a, j'\assignblock{i} a\right] \geq 1 - 2/10 = 0.8.
		\end{align*}
	\end{observation}

	The following observation will be used in the analysis:
	\begin{observation}
		\label{obs:RwC-half}
		Let $j \neq j'$ be bad jobs on $i$. Then
		\begin{align*}
			\Pr\left[\tau_{i, j} + \theta_{i, j} < \tau_{i, j'} + \theta_{i, j'} \big| j \samegroup{i} j'\right] = 1/2.
		\end{align*}
	\end{observation}
	\begin{proof}
		Fix a basic block $(2^a, 2^{a+1}]$ such that $(2^a, 2^{a+1}] \in (\bone \phi_{i, j}, p_{i, j}]$ and $(2^a, 2^{a+1}] \in (\bone \phi_{i, j'}, p_{i, j'}]$ (i.e, Property~\ref{property:RwC-job-2-block-sub-interval} for both $j$ and $j'$). 
		
		Fix a $s_{i, j}$ such that $s_{i, j} + \theta_{i, j} \leq 2^a$ (i.e, Property~\ref{property:RwC-job-2-block-stheta-small}). Condition on this $s_{i, j}$, $\tau_{i, j}$ is uniformly distributed in $(s_{i, j}, s_{i, j} + p_{i, j}]$, and thus $\tau_{i, j} + \theta_{i, j}$ is uniformly distributed in $(s_{i, j} + \theta_{i, j}, s_{i, j} + p_{i,j} + \theta_{i, j}] \supseteq (2^a, 2^{a+1}]$. Thus, conditioned on $s_{i, j}$ and $j \assignblock{i} a$,  $\tau_{i, j} + \theta_{i, j}$ is uniformly distributed in $(2^a, 2^{a+1}]$. This holds even if  we only condition on $j \assignblock{i} a$,  since the statement holds for any $s_{i, j}$ satisfying $s_{i, j} + \theta_{i, j} \leq 2^a$.
		
		The same holds for $(i, j')$.  For simplicity, denote by $e$ the event $\tau_{i, j} + \theta_{i, j} < \tau_{i, j'} + \theta_{i, j'}$. Thus, $\Pr\left[e\big|j \assignblock{i} a, j' \assignblock{i} a\right] = 1/2$ (notice that the events $j \assignblock{i} a$ and $j' \assignblock{i} a$ are independent).  Conditioned on $j \assignblock{i} a, j' \assignblock{i} a$,  the event $j \samegroup{i} j'$ is independent of the event $e$. Thus, $\Pr\left[e\big| j \samegroup{i} j', j \assignblock{i} a, j' \assignblock{i} a\right] = 1/2$.  This holds for every $a$; thus, $\Pr\left[e\big| j \samegroup{i} j'\right] = 1/2$.
	\end{proof}
	
	We define $h_{i,j}(\tau) = \sum_{s \in [\tau - p_{i,j}, \tau)} x_{i, j, s}$ to be the total height of rectangles for $j$ on $i$ that cover $\tau$.  It is easy to see that $\frac{h_{i, j}}{y_{i, j}p_{i,j}}$ is the probability density function for $\tau_{i, j}$.  Let  $A_{i, J'}(\tau) = \sum_{j \in J'}\int_{\tau' = 0}^\tau h_{i,j}(\tau') \sfd \tau'$ be the total area of the parts of rectangles $\set{R_{i, j, s}}_{j \in J', s}$ that are before $\tau$; we simply use $A_{i,j}(\tau)$ for $A_{i, \{j\}}(\tau)$. Notice that $A_{i, J}(\tau) \leq \tau$ for every $i$ and $\tau \in [0, T]$. 
	
	It is also convenient to define a set of ``shifted'' rectangles.  For every $i, j, s$, we define $R'_{i, j,s}$ to be the rectangle $R_{i, j, s}$ shifted by $\cone(s + \phi_{i,j}) + \ctwo y_{i, j} p_{i, j}$ units of time to the right. That is, $R'_{i, j, s}$ is the rectangle of height $x_{i, j, s}$ with horizontal span $(\cthree s + \cone\phi_{i, j} + \ctwo y_{i, j}p_{i,j}, \cthree s + \cone\phi_{i, j} +\ctwo y_{i,j}p_{i,j} + p_{i, j}]$.\footnote{We slightly increase $T$ so that even after shifting, all rectangles are inside the interval $(0, T]$.} Notice that $\cone (s + \phi_{i, j}) + \ctwo y_{i, j}p_{i, j}$ is the definition of $\theta_{i,j}$ if $s_{i, j} = s$.  To distinguish these rectangles from the original rectangles, we call the new rectangles $R'$-rectangles and the original ones $R$-rectangles. 
	
	Similarly, we define $h'_{i, j}(\tau)$ to be the total height of $R'$-rectangles for $j$ on $i$ that covers $\tau$. Thus, $\frac{h'_{i,j}}{y_{i, j}p_{i, j}}$ is the PDF for the random variable $\tau_{i, j} + \theta_{i, j}$. Let $A'_{i, j} (\tau) = \int_{\tau' = 0}^\tau h'_{i, j}(\tau')\sfd \tau'$ to be the total area of the parts of rectangles $\set{R'_{i, j, s}}_{j \in J', s}$ that are before $\tau$. Notice that for every $i \in M$ and $\tau \in [0, T]$, $A'_{i, J}(\tau) \leq A_{i, J}(\tau) \leq \tau$ since we only shift rectangles to the right. 
	
	For two functions $f:[0, T] \to \R_{\geq 0}$ and $F:[0, T] \to \R_{\geq 0}$, define
	\begin{align*}
	f \otimes F = \int_{\tau = 0}^T f(\tau) F(\tau) \sfd \tau.
	\end{align*}
	
	\paragraph{Bounding Expected Completion Time Job by Job} To analyze the approximation ratio of the algorithm, we fix a job $j^*$ and a machine $i$ such that $y_{i, j^*} > 0$.  We shall bound $\E\left[\tilde C_{j^*}|j^* \assignmachine i\right]$. It suffices to bound it by $\left(1.5 -\frac{1}{6000}\right)\sum_{s}\frac{x_{i, j^*, s}}{y_{i,j^*}}(s + p_{i, j^*}) = \left(1.5 -\frac{1}{6000}\right) (\phi_{i, j^*} + p_{i, j^*})$.  The following lemma gives a comprehensive upper bound that takes all parameters into account: 
	\begin{lemma} Let $I:[0,T] \to [0,T]$ be the identity function. Then $\E\left[\tilde C_{j^*} \big| j^* \assignmachine i\right]$ is at most
		\label{lemma:RwC-Bound-on-tildeC}
		\begin{align}
			\cfour \phi_{i, j^*} + \left(1.5- \cfive  y_{i,j^*}\right)p_{i,j^*} - \frac{h'_{i, j^*}}{y_{i,j^*}p_{i,j^*}}\otimes(I - A'_{i, J}) - \frac{\zeta}{2} \sum_{j \neq j^*}\Pr\left[ {j \samegroup{i} j^*}\right]y_{i,j}p_{i, j}. \label{equ:RwC-Bound-on-tildeC}
		\end{align}	
	\end{lemma}
	
	If we throw away all the negative terms, then we get an $\cfour  \phi_{i, j^*} + 1.5p_{i, j^*}$ upper bound on $\E\left[\tilde C_{j^*} \big| j^* \assignmachine i\right]$, which is at most $1.5(\phi_{i, j^*} + p_{i, j^*})$. If either $\phi_{i, j^*}$ or $y_{i, j^*}$ is large, then we can prove a better than 1.5 factor; this coincides with our definition of good jobs. For a bad job $j^*$ on $i$, we shall show that the absolute value of the third and fourth term in \eqref{equ:RwC-Bound-on-tildeC} is large. The third term is large if many rectangles are shifted by a large amount. The fourth term is where we use the strong negative correlation of Theorem~\ref{thm:BSS} from \cite{BSS16} to reduce the final approximation ratio.  
	
	\begin{proof}[Proof of Lemma~\ref{lemma:RwC-Bound-on-tildeC}]
		For notational convenience, let $e_j$ denote the event that $\tau_{i, j} + \theta_{i, j} \leq \tau_{i, j^*} + \theta_{i, j^*}$, for every $j \in J$.
	
		\begin{align}
			 \E\left[\tilde C_{j^*}\big|j^* \assignmachine i\right] &= \frac{1}{\Pr[j^* \assignmachine i]}\E\left[\mathbf{1}_{j^* \assignmachine i}\times \tilde C_{j^*}\right] \quad=\quad \frac{1}{y_{i, j^*}}\sum_{j} \Pr\left[j^* \assignmachine i, j \assignmachine i, e_j  \right]\times p_{i, j} \nonumber\\
			&= \frac{1}{y_{i, j^*}}\sum_{j \neq j^*}\Big(\Pr\left[e_j, {j \samegroup{i} j^*}\right]\times 
			\Pr\left[ j^* \assignmachine i, j \assignmachine i\big | e_j, {j \samegroup{i} j^*}\right] \nonumber \\
			&\hspace*{0.15\textwidth}+ \quad \Pr\left[e_j, {j \notsamegroup{i} j^*}\right]\times 
			\Pr\left[ j^* \assignmachine i, j \assignmachine i\big | e_j, j\notsamegroup{i} j^* \right]\Big) \times p_{i, j} + p_{i, j^*} \nonumber\\
			&\leq \frac{1}{y_{i, j^*}}\sum_{j \neq j^*} \Big(\Pr\left[e_j, {j \samegroup{i} j^*}\right](1-\zeta)y_{i, j^*}y_{i,j} + \Pr\left[e_j, {j \notsamegroup{i} j^*}\right]y_{i, j^*}y_{i,j}\Big)p_{i, j} + p_{i, j^*}\nonumber\\
			&= \sum_{j \neq j^*}\Pr\left[e_j\right]y_{i,j}p_{i,j} - \zeta \sum_{j \neq j^*}\Pr\left[e_j, {j \samegroup{i} j^*}\right]y_{i,j}p_{i, j} + p_{i, j^*}\nonumber\\
			&= \sum_{j \neq j^*}\Pr\left[e_j\right]y_{i,j}p_{i,j} - \frac{\zeta}{2} \sum_{j \neq j^*}\Pr\left[ {j \samegroup{i} j^*}\right]y_{i,j}p_{i, j} + p_{i, j^*}. \label{equ:RwC-break-bound}
		\end{align}
		The second equality is due to the fact that $\tilde C_{j^*} = \sum_{j \assignmachine i: e_j} p_{i, j}$, conditioned on $j^* \assignmachine i$. The only inequality is due to the third property of $E^*$ in Theorem~\ref{thm:BSS}: conditioned on $j \samegroup{i} j^*$ ($j \notsamegroup{i} j^*$ resp.), the probability that $j^* \assignmachine i, j \assignmachine i$ is at most $(1-\zeta)y_{i, j^*}y_{i, j}$  ($y_{i, j^*}y_{i, j}$ resp.), independent of the $\tau$ and $\theta$ values (thus, independent of $e_j$). The last equality is due to Observation~\ref{obs:RwC-half}.
		
		We focus on the first term of \eqref{equ:RwC-break-bound}:
		\begin{align}
			\quad \sum_{j \neq j^*}\Pr[e_j]y_{i, j}p_{i,j} &= \sum_{j \neq j^*} \Pr\left[\tau_{i, j} + \theta_{i,j} \leq \tau_{i, j^*} + \theta_{i, j^*}\right] y_{i, j} p_{i, j}
			= \frac{h'_{i, j^*}}{y_{i, j^*}p_{i, j^*}} \otimes A'_{i, J \setminus j^*}  \nonumber \\
			&= \frac{h'_{i,j^*}}{y_{i,j^*}p_{i,j^*}} \otimes I - \frac{h'_{i,j^*}}{y_{i,j^*}p_{i,j^*}} \otimes (I - A'_{i, J})  - \frac{h'_{i, j^*}}{y_{i,j^*}p_{i, j^*}} \otimes A'_{i, j^*}.  \label{equ:RwC-break-again}
		\end{align}
		To see the second equality, we notice that $\tau_{i, j^*} + \theta_{i, j^*}$ has PDF $\frac{h'_{i, j^*}}{y_{i, j^*}p_{i, j^*}}$, $\tau_{i, j} + \theta_{i, j}$ has PDF $\frac{h'_{i, j}}{y_{i, j^*}p_{i, j}}$ and the two random quantities are independent if $j \neq j^*$. Thus, for a fixed $\tau_{i, j^*} + \theta_{i, j^*} = \tau$, the probability that $\tau_{i, j} + \theta_{i,j}\leq\tau$ is exactly $\frac{A'_{i,j}(\tau)}{y_{i,j}p_{i,j}}$; thus contribution of $j$ is exactly $A'_{i, j}(\tau)$. Summing up over all $j \neq j^*$ gives the equality. The third equality is by $A'_{i, J \setminus {j^*}} \equiv I - (I - A'_{i, J}) - A'_{i, j^*}$. 
	
		The first term of \eqref{equ:RwC-break-again} is
		\begin{align}
			&\quad \frac{h'_{i, j^*}}{y_{i,j^*}p_{i, j^*}} \otimes I =  \sum_{s}\frac{x_{i, j^*, s}}{y_{i, j^*}}\int_{\tau = s}^{s+p_{i,j^*}} \frac{(\tau+\cone(s+\phi_{i,j^*}) + \ctwo y_{i,j^*}p_{i, j^*} ) \sfd\tau}{p_{i,j^*}}  \nonumber \\
			&= \sum_{s}\frac{x_{i, j^*, s}}{y_{i, j^*}}(s + 0.5 p_{i,j^*} + \cone(s+\phi_{i,j^*}) + \ctwo y_{i, j^*}p_{i, j^*}) = \cfour \phi_{i, j^*} + 0.5p_{i,j^*} + \ctwo y_{i, j^*}p_{i, j^*}. \nonumber
		\end{align}
			
		Now focus on the third term of \eqref{equ:RwC-break-again} :
		\begin{align}
			\frac{h'_{i, j^*}}{y_{i,j^*}p_{i,j^*}} \otimes A'_{i, j^*} &= y_{i,j^*}p_{i,j^*}\int_{\tau = 0}^T\frac{h'_{i, j^*}(\tau)}{y_{i,j^*}p_{i,j^*}}\int_{\tau' = 0}^T\frac{h'_{i, j^*}(\tau')}{y_{i,j^*}p_{i,j^*}} \mathbf{1}_{\tau' < \tau} \sfd\tau\sfd\tau'
			=\frac{y_{i,j^*}p_{i,j^*}}{2}.
		\end{align} 
		
		The first equality holds since $A'_{i, j^*}$ is the integral of $h'_{i, j^*}$. The second equality holds since the probability that $\tau' < \tau$, where $\tau$ and $\tau'$ are i.i.d random variables and are not equal almost surely, is 1/2. 
		
		Thus, by applying the above two equalities to \eqref{equ:RwC-break-again}, we have
		\begin{align*}
			\sum_{j \neq j^*}\Pr\left[e_j\right]y_{i, j}p_{i,j} = \cfour \phi_{i, j^*} + 0.5p_{i,j^*} - \frac{h'_{i, j^*}}{y_{i,j^*}p_{i,j^*}}\otimes(I - A'_{i, J}) - \cfive  y_{i, j^*}p_{i, j^*}.
		\end{align*}
		
		Applying the above equality to \eqref{equ:RwC-break-bound}, we that $\E\left[\tilde C_{j^*}|j^* \assignmachine i\right]$ is at most
		\begin{align*}
			\cfour \phi_{i, j^*} + \left(1.5-\cfive  y_{i,j^*}\right)p_{i,j^*} - \frac{h'_{i, j^*}}{y_{i,j^*}p_{i,j^*}}\otimes(I - A'_{i, J}) - \frac{\zeta}{2} \sum_{j \neq j^*}\Pr\left[ {j \samegroup{i} j^*}\right]y_{i,j}p_{i, j}.
		\end{align*}
		This is exactly \eqref{equ:RwC-Bound-on-tildeC}.
	\end{proof}
	
	With the lemma, we shall analyze good jobs and bad jobs separately. For good jobs, the bound follows directly from the definition:
	
	\begin{lemma}
		\label{lemma:RwC-bound-good}
		If $j^*$ is good on $i$, then $\displaystyle \E\left[\tilde C_{j^*}\big| j^* \assignmachine i\right]\leq 1.4991(\phi_{i, j^*} + p_{i, j^*})$.
	\end{lemma}
	\begin{proof}
		We have $\displaystyle \E\left[\tilde C_{j^*}\big|j^* \assignmachine i\right] \leq \cfour \phi_{i, j^*} + \left(1.5-\cfive y_{i,j^*}\right)p_{i,j^*}$, by throwing the two negative terms in \eqref{equ:RwC-Bound-on-tildeC}.  Since $j^*$ is good on $i$,  we have $\phi_{i, j^*} + y_{i, j^*}p_{i, j^*} \geq \csix  p_{i, j^*}$.
		\begin{flalign*}
			&& \quad \E\left[\tilde C_{j^*}\big|j^* \assignmachine i\right]  &\leq 1.5 (\phi_{i, j^*} + p_{i, j^*}) - (\cfive \phi_{i, j} + \cfive  y_{i, j^*}p_{i, j^*}) \\ 
			&& &\leq 1.5(\phi_{i, j^*} + p_{i, j^*}) - 0.01\phi_{i, j^*} - 0.09(\phi_{i, j^*}+y_{i, j^*}p_{i,j^*})\\
			&& &\leq 1.5(\phi_{i, j^*} + p_{i, j^*}) - 0.01\phi_{i, j^*} - 0.09 \times \csix p_{i, j^*} \leq 1.4991(\phi_{i, j^*} + p_{i, j^*}). && \qedhere
		\end{flalign*}
	\end{proof}
	
	Thus, it remains to consider the case where $j^*$ is bad on $i$. The rest of the section is devoted to the proof of the following lemma:
	\begin{lemma}
		\label{lemma:RwC-bound-for-bad}
		If $j^*$ is bad on $i$, then $\E\left[\tilde C_{j^*} \big| j^* \assignmachine i \right] \leq \left(1.5-\frac{1}{6000}\right) (\phi_{i,j^*} + p_{i, j^*})$.
	\end{lemma}

	It suffices to give a lower bound on the sum of the absolute values of negative terms in \eqref{equ:RwC-Bound-on-tildeC}: 
	\begin{align}
		\cfive  y_{i, j^*}p_{i, j^*} + \frac{h'_{i, j^*}}{y_{i,j^*}p_{i,j^*}}\otimes(I- A'_{i, J}) + \frac{\zeta}{2} \sum_{j \neq j^*}\Pr\left[ {j \samegroup{i} j^*}\right]y_{i,j}p_{i, j} \geq \frac{p_{i, j^*}}{6000}. \label{inequ:RwC-lower-bound-neg}
	\end{align}
	If the above inequality holds, then we have $\E\left[\tilde C_{j^*} \big| j^* \assignmachine i \right] \leq \cfour  \phi_{i, j^*} + \left(1.5-\frac{1}{6000}\right)p_{i, j^*} \leq \left(1.5-\frac{1}{6000}\right) \allowbreak (\phi_{i, j^*} + p_{i, j^*})$, implying Lemma~\ref{lemma:RwC-bound-for-bad}.

	To prove \eqref{inequ:RwC-lower-bound-neg}, we construct a set of \emph{configurations} with total weight 1, and lower bound the left-side configuration by configuration. We define a configuration $U$ to be a set of pairs in $J \times \set{0, 1, 2, \cdots, T-1}$ such that for every two distinct pairs $(j, s), (j', s') \in U$, the two intervals $(s, s+p_{i,j}]$ and $(s', s'+p_{i, j'}]$ are disjoint.\footnote{Notice that this configuration does not necessarily correspond to a valid scheduling on $i$, since it may contain two pairs with the same $j$.} For the sake of the description, we also view $(j, s)$ as the interval $(s, s+p_{i, j}]$ associated with the job $j$.
	
	Recall that the total height of all $R$-rectangles on $i$ covering any time point is at most $1$.   It is a folklore result that we can find a set of configurations, each configuration $U$ with a $z_U > 0$, such that $\sum_{U}z_U = 1$ and $\sum_{U \ni (s, j)}z_U = x_{i, j, s}$ for every $j$ and $s$.
	
	
	With the decomposition of the $R$-rectangles on $i$ into a convex combination of configurations, we can now analyze the contribution of each configuration to the left of \eqref{inequ:RwC-lower-bound-neg}.  For any configuration $U$, we define a function $\ell_U: [0, T] \to \R_{\geq 0}$ as follows:
	
	\begin{align*}
		\tilde \ell_U(\tau) = \sum_{(j, s) \in U: \cthree s + \cone\phi_{i,j} + \ctwo y_{i, j} p_{i, j} \leq \tau} \min\set{\tau - (\cthree s + \cone\phi_{i, j} + \ctwo y_{i, j}p_{i,j}), p_{i, j}}.
	\end{align*}
	
	The definition comes from the following process. Focus on the intervals $\set{(s, s+p_{i,j}): (j, s) \in U}$. We then shift each interval $(s, s + p_{i, j}]$ to the right by $\cone s + \cone\phi_{i, j} + \ctwo y_{i,j}p_{i,j}$; notice that this is exactly the definition of $\theta_{i, j}$ when $s_{i, j} = s$.  Then $\ell_U(\tau)$ is exactly the total length of the sub-intervals of the shifted intervals before time point $\tau$. Recalling that $A'_{i, J}(\tau)$ is the total area of the parts of the $R'$-rectangles on $i$ before time point $\tau$, and each $R'_{i, j, s}$ is obtained by shifting $R_{i, j, s}$ by $\cone s + \cone\phi_{i, j} + \ctwo y_{i, j}p_{i,j}$ to the right, the following holds: 
	\begin{align}
		A'_{i, J} \equiv \sum_{U}z_U \ell_U. \label{equ:RwC-break-Ap}
	\end{align}
	
	Let $c_U(j) = \big|\{s: (j, s) \in U\}\big|$ be the number of pairs in $U$ for the job $j$.  Now, we can define the contribution of $U$ to the bound to be
	\begin{align*}
		D_U := z_U\left(\cfive p_{i, j^*}c_U(j^*) + \frac{h'_{i, j^*}}{y_{i, j^*}p_{i, j^*}}\otimes \big(I  - \ell_U\big) + \frac{\zeta}{2}\sum_{j \neq j^*}\Pr[j \samegroup{i} j^*]p_{i, j}c_U(j) \right).
	\end{align*}
	
	\begin{claim}
		\label{claim:RwC-break-contribution}
		The left side of \eqref{inequ:RwC-lower-bound-neg} is exactly $\sum_{U}D_U$.
	\end{claim}
	\begin{proof}
		Indeed, the three terms in \eqref{inequ:RwC-lower-bound-neg} is respectively the sum over all $U$ of each of the three terms in the definition of $D_U$.  For the first and the third term, the equality comes from $y_{i, j} = \sum_{U}z_U c_U(j)$ for every $j$. The equality for the second term comes from $\sum_U z_U(I - \ell_U) = I - A'_{i, J}$.
	\end{proof}

	\paragraph{Lower Bound the Contribution of Each $U$} Now we fix some configuration $U$ such that $z_U > 0$. Let $a$ be the largest integer such that $2^{a+1} \leq 0.9 p_{i, j^*}$; thus $a\geq -2$. We can focus on the basic block $(2^a, 2^{a+1}]$. Notice that the length of the basic block is  $2^a \geq 0.9 p_{i, j^*}/4$, by the definition of $a$.  
	
	The following simple observations are useful in establishing our bounds.	
		
	\begin{observation}
		\label{obs:RwC-bound-ell-0}		
		If $s \leq 18 \phi_{i, j^*}$, then $R'_{i, j^*, s}$ will cover $(2^a, p_{i, j^*}]$. 
	\end{observation}
	\begin{proof}		
		The rectangle $R'_{i, j^*, s}$ will cover $(2^a, p_{i, j^*}]$ if $s + \cone(s + \phi_{i, j^*}) + \ctwo y_{i, j}p_{i,j^*} \leq 2^a$, which is $s \leq (2^a-\cone\phi_{i, j^*} - \ctwo y_{i, j}p_{i,j^*})/\cthree $.  Since $j^*$ is bad on $i$, we have $0.2 \phi_{i, j^*} + 0.4 y_{i, j^*}p_{i, j^*} \leq 0.4 \times \csix p_{i, j^*} = 0.004 p_{i, j^*}$. Thus, $(2^a-\cone\phi_{i, j^*} - \ctwo y_{i, j}p_{i,j^*})/\cthree  \geq (2^a-0.004 p_{i, j^*})/\cthree \geq (0.9/4 -0.004)p_{i,j^*}/1.2 \geq \frac{0.9/4-0.004}{1.2}\frac{1}{\csix}\phi_{i, j^*} \geq 18\phi_{i, j^*}$. Thus, if $s \leq 18\phi_{i, j^*}$, we have $s + \cone(s + \phi_{i, j^*}) + \ctwo y_{i, j}p_{i,j^*} \leq 2^a$.  
	\end{proof}
	
	\begin{observation}
		For every $\tau \in (2^a, p_{i, j^*}]$, we have $h'_{i, j^*}(\tau) \geq \frac{17}{18}y_{i, j^*}$. 
			\label{obs:RwC-bound-ell-1}
	\end{observation}
	\begin{proof}
		This comes from Observation~\ref{obs:RwC-bound-ell-0}. $R'_{i, j^*, s}$ will cover $\tau$ if $s \leq 18\phi_{i, j^*}$. 
		By the definition of $\phi_{i, j}$ and Markov inequality, the sum of $x_{i, j^*, s}$ over all such $s$ is at least $\frac{17}{18}y_{i, j^*}$. 	
	\end{proof}
	
	\begin{observation}
		\label{obs:RwC-bound-ell-2}
		If $\tau'$ is not contained in the \emph{interior} of any interval in $U$ and $\tau \geq \tau'$, then $\tau  - \ell_U(\tau) \geq \min\set{\cone \tau', \tau - \tau'}$.
	\end{observation}
	\begin{proof}
		Consider the definition of $\ell_U$ via the shifting of intervals. 
		 Since $\tau'$ is not contained in the interior of any interval in $U$, an interval in $U$ is either to the left of $\tau'$ or to the right of $\tau'$.  By the way we shifting intervals, any interval to the right of $\tau'$ will be shifted by at least $\cone\tau'$ distance to the right. Thus, the sub-intervals of the intervals in $U$ from $\max\set{\tau', \tau - \cone \tau'}$ to $\tau$ will be shifted to the right of $\tau$. The observation follows. 
	\end{proof}
	
	\begin{observation}
		\label{obs:RwC-bound-ell-3}
		If $\tau$ is covered by some interval $(s, s+p_{i, j}]$ in $U$, then $$\tau - \ell_U(\tau) \geq \min\set{\cone(s+\phi_{i,j}) + \ctwo y_{i,j} p_{i,j}, \tau - s}.$$
	\end{observation}
	\begin{proof} 
		The interval $(s, s+p_{i, j}]$ will be shifted by $\cone(s+\phi_{i, j})+\ctwo y_{i,j}p_{i,j}$ distance to the right. So the sub-interval $( \max\set{\tau - \cone(s+\phi_{i, j})-\ctwo p_{i,j}, s}, \tau]$ will be shifted to the right of $\tau$. The observation follows.
	\end{proof}
	
	Equipped with these observations, we can analyze the contribution of $U$ case by case: \bigskip
	
	\noindent {\bf Case 1:} $(2^a, 0.92 p_{i, j^*}]$ is not a sub-interval of any interval in $U$. In this case, there is $\tau' \in (2^a, 0.92 p_{i, j^*})$ that is not in the interior of any interval in $U$. By Observation~\ref{obs:RwC-bound-ell-2},  any time point in $ \tau \in (0.95p_{i, j^*},p_{i, j^*}]$ has $\tau - \ell_U(\tau) \geq \min\set{\tau - \tau', \cone\tau'} \geq \min\{0.95p_{i,j^*} - 0.92p_{i,j^*}, \cone \times 2^a\} \geq \min\set{0.03 p_{i, j^*}, \cone \times 0.9p_{i, j^*}/4} = 0.03 p_{i, j^*}$. By Observation~\ref{obs:RwC-bound-ell-1}, any such $\tau$ has $h'$ value at least $\frac{17}{18} y_{i, j^*}$. Thus,  the contribution of $U$ is 
	\begin{align*}
		D_U &\geq z_U \times \frac{h'_{i, j^*}}{y_{i, j^*}p_{i,j^*}}\otimes (I-\ell_U) \geq z_U \times \frac{17y_{i, j^*}}{18y_{i,j^*}p_{i,j^*}} \times (0.03 p_{i, j^*})\times (0.05 p_{i, j^*})\\
		&\geq 0.0014z_Up_{i,j^*} \geq \frac{z_U p_{i, j^*}}{6000}.
	\end{align*}
%

		\noindent {\bf Case 2:} $(2^a, 0.92 p_{i, j^*}]$ is covered by some interval $(s, s+p_{i, j}]$ in $U$, and $\cone (\phi_{i, j} + s) + \ctwo y_{i, j}p_{i,j} \geq 0.002  \times 2^a$. Focus on any $\tau \in (1.01 \times 2^a, 0.92 p_{i, j^*}] \subseteq (2^a,  p_{i, j^*}]$. By Observation~\ref{obs:RwC-bound-ell-1}, we have $h'(\tau) \geq \frac{17}{18} y_{i, j^*}$. By Observation~\ref{obs:RwC-bound-ell-3}, we have $\tau - \ell_U(\tau) \geq \min\{\cone (\phi_{i, j} + s) + \ctwo y_{i, j}p_{i,j}, \tau - s\} \geq  \min\set{0.002 \times 2^a, 0.01\times 2^a} = 0.002\times 2^a$. Thus, the contribution of $U$ is 
	\begin{align*}
		D_U &\geq z_U \times \frac{h'_{i, j^*}}{y_{i, j^*}p_{i,j^*}}\otimes (I-\ell_U)  \geq z_U \times \frac{17y_{i,j^*}}{18y_{i,j^*}p_{i,j^*}} \times (0.002\times 2^a) \times (0.92 p_{i, j^*} - 1.01 \times 2^a) \\
		& \geq \left(\frac{17}{18} \times 0.002 \times 0.9/4 \times (0.92-1.01\times 0.9/2)\right)z_U p_{i, j^*} \geq 0.00019 z_U p_{i, j^*} \geq \frac{z_U p_{i, j^*}}{6000},
	\end{align*}		
 where we used $0.9p_{i, j^*}/4 \leq 2^a \leq 0.9p_{i, j^*}/2$. \medskip
		
		\noindent {\bf Case 3:} $(2^a, 0.92 p_{i, j^*}]$ is covered by some interval $(s, s+p_{i, j}]$ in $U$, and $\cone (\phi_{i, j} + s) + \ctwo y_{i, j}p_{i,j} < 0.002 \times 2^a$.  If $j = j^*$, then $D_U \geq z_U \times \cfive p_{i, j^*}c_U(j^*) \geq z_U \times \cfive  p_{i, j^*}$. So, we assume $j \neq j^*$.  This is where we use the strong negative correlation between $j$ and $j^*$.  We shall lower bound $\Pr[j^* \assignblock{i} a]$ and $\Pr[j \assignblock{i} a]$ separately. 
		
		Notice that $2^a \geq 0.9p_{i, j^*}/4 \geq  0.9 \times \frac{1}{\csix} \phi_{i, j^*} /4 \geq 10 \phi_{i, j^*}$, by the fact that $j^*$ is bad on $i$. Thus, $(2^a, 2^{a+1}] \subseteq (\bone\phi_{i, j^*}, p_{i, j^*}]$, implying Property~\ref{property:RwC-job-2-block-sub-interval} for $j^*$.  If $s_{i, j^*} = s$ and $R'_{i, j^*, s}$ covers $(2^a, 2^{a+1}]$, then Property~\ref{property:RwC-job-2-block-stheta-small} holds. 
		By Observation~\ref{obs:RwC-bound-ell-0}, this happens with probability at least $\frac{17}{18}$. Thus, we have
		\begin{align*}
			\Pr\left[j^* \assignblock{i} a\right] \geq \frac{17}{18} \cdot \frac{2^a}{p_{i, j^*}} \geq \frac{17}{18}\cdot \frac{0.9}{4} \geq 0.21.
		\end{align*}
		
	
	Now, we continue to bound $\Pr\left[j \assignblock{i} a\right]$. To do this, we need to first prove that $j$ is bad on $i$. Indeed, $\phi_{i, j} + y_{i,j}p_{i,j} \leq 5(\cone (\phi_{i, j} + s) + \ctwo y_{i, j}p_{i,j}) < 5\times 0.002 \times 2^a \leq \csix p_{i, j} $, since $(s, s+p_{i, j}] \supseteq (2^a, 0.92p_{i, j^*}] \supseteq (2^a, 2^{a+1}]$, which is of length at least $2^a$. This implies that $j$ is bad on $i$.
				
	Then, $s \leq 0.002\times 2^a/\cone  = 0.01 \times 2^a$ and $s + p_{i, j} > 0.92 p_{i, j^*} \geq \frac{0.92}{0.9}\times 2^{a+1} \geq 2^{a+1} + 0.01\times 2^a$.  This implies that $p_{i, j} \geq 2^{a+1}$.  Also, $2^a \geq \cone  \phi_{i, j}/0.002 = 100 \phi_{i, j} \geq \bone\phi_{i, j}$. Thus, Property~\ref{property:RwC-job-2-block-sub-interval} holds.
		
		For every $s' \leq 50\phi_{i, j}$, we have $\cthree  s' + \cone \phi_{i, j} + \ctwo y_{i, j} p_{i, j} \leq  60.2 \phi_{i, j}  + \ctwo y_{i, j}p_{i, j} \leq 
		\frac{60.2}{\cone} (\cone (\phi_{i, j} + s) + \ctwo y_{i, j}p_{i,j} ) \leq \frac{60.2}{\cone} \times 0.002\times 2^a \leq 2^a$. Thus, Property~\ref{property:RwC-job-2-block-stheta-small} holds if $s_{i, j} \leq 50\phi_{i, j}$.  Notice that $\Pr[s_{i, j} \leq 50\phi_{i, j}] \geq 0.98$. Under this condition, the probability of $j \assignblock{i} a$ is at least $\frac{2^a}{p_{i, j}} \geq \frac{0.9}{4}\frac{p_{i,j^*}}{p_{i, j}}$. Overall, $\Pr\left[j \assignblock{i} a\right] \geq 0.98 \times \frac{0.9}{4}\frac{p_{i,j^*}}{p_{i, j}} \geq 0.22 \frac{p_{i,j^*}}{p_{i, j}}$.
		
		
		Thus, by Observation~\ref{obs:RwC-prob-same-group}, 
		\begin{align*}
			\Pr[j \samegroup{i} j^*] \geq 0.8 \Pr\left[j^* \assignblock{i} a\right] \Pr\left[j \assignblock{i} a\right] \geq 0.8 \times 0.21 \times 0.22 \frac{p_{i, j^*}}{p_{i,j}} \geq 0.036\frac{p_{i, j^*}}{p_{i,j}}.
		\end{align*}
		Thus, the contribution of $U$ is $D_U \geq  z_U\frac{\zeta}{2}\Pr[j \samegroup{i} j^*]p_{i, j} \geq 0.018\zeta z_U p_{i, j^*} = \frac{z_U p_{i, j^*}}{6000}$.
		
	Thus, for every $U$, the contribution of $U$ is at least $\frac{z_Up_{i, j^*}}{6000}$. By Claim~\ref{claim:RwC-break-contribution}, the left side of \eqref{inequ:RwC-lower-bound-neg} is $\sum_{U}D_U \geq \sum_U\frac{z_Up_{i, j^*}}{6000} = \frac{p_{i, j^*}}{6000}$. This finishes the proof of Lemma~\ref{lemma:RwC-bound-for-bad}. 
	
	So, we always have $\E\left[\tilde C_{j^*}|j^* \assignmachine i\right] \leq \left(1.5-\frac{1}{6000}\right)(\phi_{i,j^*} + p_{i, j^*})$. Deconditioning on $j^* \assignmachine i$, we have $\E\left[\tilde C_{j^*}\right] \leq \left(1.5-\frac{1}{6000}\right)\sum_{i}y_{i, j^*}(\phi_{i, j^*} + p_{i, j^*}) = \left(1.5-\frac{1}{6000}\right)C_j$. This finishes the proof of Theorem~\ref{theorem:main-RwC}.

	\section{Handling Super-Polynomial $T$}
\label{sec:superT}

In this section, we show how to handle the case when $T$ is super-polynomial in $n$  
for $P|\tprec|\sum_j w_j C_j$ and $R|\tprec|\sum_j w_j C_j$. The way we handle super-polynomial $T$ is the same as that in \cite{IL16}.  \cite{IL16} considers the problem $R|r_j|\sum_j w_j C_j$, i.e, the unrelated machine job scheduling with job arrival times.  They showed how to efficiently obtain a $(1+\epsilon)$ approximate LP solution that only contains polynomial number of non-zero variables. Since the problem they considered is more general than $R||\sum_j w_j C_j$ and their LP is also a generalization of our \eqref{LP:RwC}, their technique can be directly applied to our algorithm for $R||\sum_j w_j C_j$.  Due to the precedence constraints, the technique does not  directly apply to $P|\tprec|\sum_j w_j C_j$. However, with a trivial modification, it can handle the precedence constraints as well.  We omit the detail here since the analysis will be almost identical to that in \cite{IL16}.

	\bibliographystyle{plain}
	\bibliography{reflist}

\begin{thebibliography}{10}

\bibitem{BK09}
Nikhil Bansal and Subhash Khot.
\newblock Optimal long code test with one free bit.
\newblock In {\em Proceedings of the 2009 50th Annual IEEE Symposium on
  Foundations of Computer Science}, FOCS '09, pages 453--462. IEEE Computer
  Society, 2009.

\bibitem{BK15}
Nikhil Bansal and Janardhan Kulkarni.
\newblock Minimizing flow-time on unrelated machines.
\newblock In {\em Proceedings of the Forty-seventh Annual ACM Symposium on
  Theory of Computing}, STOC '15, pages 851--860. ACM, 2015.

\bibitem{BP10}
Nikhil Bansal and Kirk Pruhs.
\newblock The geometry of scheduling.
\newblock In {\em Proceedings of the 2010 IEEE 51st Annual Symposium on
  Foundations of Computer Science}, FOCS '10, pages 407--414. IEEE Computer
  Society, 2010.

\bibitem{BSS16}
Nikhil Bansal, Aravind Srinivasan, and Ola Svensson.
\newblock Lift-and-round to improve weighted completion time on unrelated
  machines.
\newblock In {\em Proceedings of the Forty-eighth Annual ACM Symposium on
  Theory of Computing}, STOC '16, pages 156--167. ACM, 2016.

\bibitem{BN15}
Abbas Bazzi and Ashkan Norouzi-Fard.
\newblock {\em Towards Tight Lower Bounds for Scheduling Problems}, pages
  118--129.
\newblock Springer Berlin Heidelberg, 2015.

\bibitem{CPS96}
Soumen Chakrabarti, Cynthia~A. Phillips, Andreas~S. Schulz, David~B. Shmoys,
  Cliff Stein, and Joel Wein.
\newblock {\em Improved scheduling algorithms for minsum criteria}, pages
  646--657.
\newblock Springer Berlin Heidelberg, 1996.

\bibitem{CKL15}
Deeparnab Chakrabarty, Sanjeev Khanna, and Shi Li.
\newblock On $(1,\epsilon)$-restricted asgsignment makespan minimization.
\newblock In {\em Proceedings of the 26th Annual ACM-SIAM Symposium on Discrete
  Algorithms (SODA 2015)}, 2015.

\bibitem{CK04}
C.~Chekuri and S.~Khanna.
\newblock {Approximation algorithms for minimizing average weighted completion
  time}.
\newblock {\em Handbook of Scheduling: Algorithms, Models, and Performance
  Analysis. CRC Press, Inc., Boca Raton, FL, USA}, 2004.

\bibitem{CMN01}
C.~Chekuri, R.~Motwani, B.~Natarajan, and C.~Stein.
\newblock Approximation techniques for average completion time scheduling.
\newblock {\em SIAM J. Comput.}, 31(1):146--166, January 2002.

\bibitem{CK01}
Chandra Chekuri and Sanjeev Khanna.
\newblock {\em A PTAS for Minimizing Weighted Completion Time on Uniformly
  Related Machines}, pages 848--861.
\newblock Springer Berlin Heidelberg, 2001.

\bibitem{CK02}
Chandra Chekuri and Sanjeev Khanna.
\newblock Approximation schemes for preemptive weighted flow time.
\newblock In {\em Proceedings of the Thiry-fourth Annual ACM Symposium on
  Theory of Computing}, STOC '02, pages 297--305. ACM, 2002.

\bibitem{CS97}
Fabi\'{a}n~A. Chudak and David~B. Shmoys.
\newblock Approximation algorithms for precedence-constrained scheduling
  problems on parallel machines that run at different speeds.
\newblock In {\em Proceedings of the Eighth Annual ACM-SIAM Symposium on
  Discrete Algorithms}, SODA '97, pages 581--590. Society for Industrial and
  Applied Mathematics, 1997.

\bibitem{EKS08}
Tom\'{a}\v{s} Ebenlendr, Marek Kr\v{c}\'{a}l, and Ji\v{r}\'{\i} Sgall.
\newblock Graph balancing: A special case of scheduling unrelated parallel
  machines.
\newblock In {\em Proceedings of the Nineteenth Annual ACM-SIAM Symposium on
  Discrete Algorithms}, SODA '08, pages 483--490. Society for Industrial and
  Applied Mathematics, 2008.

\bibitem{GK07}
Naveen Garg and Amit Kumar.
\newblock Minimizing average flow-time: Upper and lower bounds.
\newblock In {\em Proceedings of the 48th Annual IEEE Symposium on Foundations
  of Computer Science}, FOCS '07, pages 603--613. IEEE Computer Society, 2007.

\bibitem{Goe97}
Michel~X. Goemans.
\newblock Improved approximation algorthims for scheduling with release dates.
\newblock In {\em Proceedings of the Eighth Annual ACM-SIAM Symposium on
  Discrete Algorithms}, SODA '97, pages 591--598. Society for Industrial and
  Applied Mathematics, 1997.

\bibitem{Gra69}
R.~L. Graham.
\newblock Bounds on multiprocessing timing anomalies.
\newblock {\em SIAM JOURNAL ON APPLIED MATHEMATICS}, 17(2):416--429, 1969.

\bibitem{GLLR79}
R.~L. Graham, E.~L. Lawler, J.~K. Lenstra, and A.~H. G.~Rinnooy Kan.
\newblock {Optimization and approximation in deterministic sequencing and
  scheduling: a survey}.
\newblock {\em Ann. Discrete Math.}, 4:287--326, 1979.

\bibitem{HSS97}
Leslie~A. Hall, Andreas~S. Schulz, David~B. Shmoys, and Joel Wein.
\newblock Scheduling to minimize average completion time: Off-line and on-line
  approximation algorithms.
\newblock {\em Math. Oper. Res.}, 22(3):513--544, August 1997.

\bibitem{IL16}
Sungjin Im and Shi Li.
\newblock Better unrelated machine scheduling for weighted completion time via
  random offsets from non-uniform distributions.
\newblock In {\em {IEEE} 57th Annual Symposium on Foundations of Computer
  Science, {FOCS} 2016}, pages 138--147, 2016.

\bibitem{Jaf80}
Jeffrey~M. Jaffe.
\newblock Efficient scheduling of tasks without full use of processor
  resources.
\newblock {\em Theoretical Computer Science}, 12(1):1 -- 17, 1980.

\bibitem{KL17}
Klaus Jansen and Lars Rohwedder.
\newblock On the configuration-{LP} of the restricted assignment problem.
\newblock In {\em Proceedings of the Twenty-Eighth Annual ACM-SIAM Symposium on
  Discrete Algorithms}, SODA '17, pages 2670--2678. Society for Industrial and
  Applied Mathematics, 2017.

\bibitem{KMP08}
VS~Anil Kumar, Madhav~V Marathe, Srinivasan Parthasarathy, and Aravind
  Srinivasan.
\newblock Minimum weighted completion time.
\newblock In {\em Encyclopedia of Algorithms}, pages 544--546. Springer, 2008.

\bibitem{LR78}
J.~K. Lenstra and A.~H.~G. Rinnooy~Kan.
\newblock Complexity of scheduling under precedence constraints.
\newblock {\em Oper. Res.}, 26(1):22--35, February 1978.

\bibitem{LST90}
J.~K. Lenstra, D.~B. Shmoys, and \'{E}. Tardos.
\newblock Approximation algorithms for scheduling unrelated parallel machines.
\newblock {\em Math. Program.}, 46(3):259--271, February 1990.

\bibitem{LR97}
Stefano Leonardi and Danny Raz.
\newblock Approximating total flow time on parallel machines.
\newblock In {\em Proceedings of the Twenty-ninth Annual ACM Symposium on
  Theory of Computing}, STOC '97, pages 110--119. ACM, 1997.

\bibitem{LR16}
Elaine Levey and Thomas Rothvoss.
\newblock A (1+epsilon)-approximation for makespan scheduling with precedence
  constraints using {LP} hierarchies.
\newblock In {\em Proceedings of the Forty-eighth Annual ACM Symposium on
  Theory of Computing}, STOC '16, pages 168--177. ACM, 2016.

\bibitem{MQS98}
Alix Munier, Maurice Queyranne, and Andreas~S. Schulz.
\newblock {\em Approximation Bounds for a General Class of Precedence
  Constrained Parallel Machine Scheduling Problems}, pages 367--382.
\newblock Springer Berlin Heidelberg, 1998.

\bibitem{PSW98}
Cynthia Phillips, Clifford Stein, and Joel Wein.
\newblock Minimizing average completion time in the presence of release dates.
\newblock {\em Mathematical Programming}, 82(1):199--223, Jun 1998.

\bibitem{QS06}
Maurice Queyranne and Andreas~S. Schulz.
\newblock Approximation bounds for a general class of precedence constrained
  parallel machine scheduling problems.
\newblock {\em SIAM J. Comput.}, 35(5):1241--1253, May 2006.

\bibitem{QS02}
Maurice Queyranne and Maxim Sviridenko.
\newblock Approximation algorithms for shop scheduling problems with minsum
  objective.
\newblock {\em Journal of Scheduling}, 5(4):287--305, 2002.

\bibitem{SS97a}
Andreas~S. Schulz and Martin Skutella.
\newblock Random-based scheduling: New approximations and {LP} lower bounds.
\newblock In {\em Proceedings of the International Workshop on Randomization
  and Approximation Techniques in Computer Science}, RANDOM '97, pages
  119--133. Springer-Verlag, 1997.

\bibitem{SS02}
Andreas~S. Schulz and Martin Skutella.
\newblock Scheduling unrelated machines by randomized rounding.
\newblock {\em SIAM J. Discret. Math.}, 15(4):450--469, April 2002.

\bibitem{SW99a}
Petra Schuurman and Gerhard~J. Woeginger.
\newblock Polynomial time approximation algorithms for machine scheduling: Ten
  open problems, 1999.

\bibitem{SS99}
Jay Sethuraman and Mark~S. Squillante.
\newblock Optimal scheduling of multiclass parallel machines.
\newblock In {\em Proceedings of the Tenth Annual ACM-SIAM Symposium on
  Discrete Algorithms}, SODA '99, pages 963--964. Society for Industrial and
  Applied Mathematics, 1999.

\bibitem{Sit09}
Ren{\'e}~A. Sitters.
\newblock Approximation and online algorithms.
\newblock chapter Minimizing Average Flow Time on Unrelated Machines, pages
  67--77. Springer-Verlag, 2009.

\bibitem{Sku01}
Martin Skutella.
\newblock Convex quadratic and semidefinite programming relaxations in
  scheduling.
\newblock {\em J. ACM}, 48(2):206--242, March 2001.

\bibitem{Sku16}
Martin Skutella.
\newblock A 2.542-approximation for precedence constrained single machine
  scheduling with release dates and total weighted completion time objective.
\newblock {\em Operations Research Letters}, 44(5):676 -- 679, 2016.

\bibitem{SW99c}
Martin Skutella and Gerhard~J. Woeginger.
\newblock A ptas for minimizing the weighted sum of job completion times on
  parallel machines.
\newblock In {\em Proceedings of the Thirty-first Annual ACM Symposium on
  Theory of Computing}, STOC '99, pages 400--407. ACM, 1999.

\bibitem{Sve10}
Ola Svensson.
\newblock Conditional hardness of precedence constrained scheduling on
  identical machines.
\newblock In {\em Proceedings of the Forty-second ACM Symposium on Theory of
  Computing}, STOC '10, pages 745--754. ACM, 2010.

\bibitem{Sve11}
Ola Svensson.
\newblock Santa claus schedules jobs on unrelated machines.
\newblock In {\em Proceedings of the Forty-third Annual ACM Symposium on Theory
  of Computing}, STOC '11, pages 617--626. ACM, 2011.

\end{thebibliography}
%
\end{document}